%% file: main_article_v3.tex
\documentclass[a4paper,12pt]{article}
\input{my-heading}

\usepackage[top=2cm,bottom=2cm,left=2cm,right=2cm]{geometry}
\numberwithin{equation}{section}
\usepackage{slashed}
\usepackage{xcolor}
\usepackage{dsfont}
\usepackage{autobreak}
\usepackage{authblk}
\usepackage{cleveref}
\usepackage{lipsum}
\renewcommand\blfootnote[1]{%
  \begingroup
  \renewcommand\thefootnote{}\footnote{#1}%
  \addtocounter{footnote}{-1}%
  \endgroup
}

\newcommand{\overbar}[1]{\mkern 1.5mu\overline{\mkern-1.5mu#1\mkern-1.5mu}\mkern 1.5mu}
\renewcommand{\wedge}{}
\newcommand{\psibar}{\overbar{\psi}}
\newcommand{\thetabar}{\overbar{\theta}}

\newcommand{\sigmabar}{\overbar{\sigma}}

\newcommand{\D}{\mathcal{D}}

\newcommand{\rhobar}{\overbar{\rho}}
\newcommand{\xitilde}{\tilde{\xi}}
\newcommand{\thetatilde}{\tilde{\theta}}
\newcommand{\xibar}{\overbar{\xi}}

\newcommand{\atheta}{\theta}
\newcommand{\stheta}{\eta}
\newcommand{\astheta}{\zeta}

\newcommand{\alphadot}{{A^\prime}}
\newcommand{\betadot}{{B^\prime}}
\newcommand{\gammadot}{{C^\prime}}
\newcommand{\deltadot}{{D^\prime}}

\newcommand{\B}{{B}}
\newcommand{\C}{{C}}

\newcommand{\Aprime}{{A'}}
\newcommand{\Bprime}{{B'}}

\newcommand{\Q}{\mathcal{Q}}
\newcommand{\Qbar}{\overbar{\Q}}
\newcommand{\gbar}{\overbar{g}}

\newcommand{\sigmahat}{\hat{\sigma}}
\newcommand{\alga}{\mathfrak{a}}
\newcommand{\dP}{\mathrm{d_P}}
\newcommand{\dT}{\mathrm{d_T}}

\title{Presymplectic BV-AKSZ for $N=1$, $D=4$ Supergravity\\
\vspace{0.5cm}}

\author[1,$\dagger$,$\ddagger$]{~~~~Maxim Grigoriev }

\author[2]{Alexander Mamekin}

\affil[1]{\textsl{Service de Physique de l'Univers, Champs et Gravitation, \protect\\ Universit\'e de Mons, 20 place du Parc, 7000 Mons, 
Belgium \vspace{5pt}}}
 
\affil[2]{\textsl{ Institute for Theoretical and Mathematical Physics,\protect\\
  Lomonosov Moscow State University, 119991 Moscow, Russia  \vspace{5pt}}}

\date{}

\begin{document}
\maketitle
\begin{abstract}
We elaborate on the presymplectic BV-AKSZ approach to supersymmetric systems. In particular, we construct such a formulation for the $N=1$, $D=4$ supergravity by taking as a target space the Chevalley-Eilenberg complex of the super-Poincar\'e algebra which, as we demonstrate, admits an invariant presymplectic structure of degree $3$. This data encodes a full-scale Batalin-Vilkovisky formulation of the system, including a concise form of the BV master action. The important feature of (presymplectic) AKSZ models is that, at least at the level of equations of motion, they can be equivalently reformulated in the spacetime obtained by adding or eliminating contractible dimensions. For instance, the presymplectic AKSZ formulation of gravity can be lifted to the respective ``group manifold'', endowing it with the structure of a principle bundle and the Cartan connection therein, at least locally. In particular, the so-called rheonomy conditions emerge as a part of the presymplectic AKSZ equations of motion. The analogous considerations apply to supergravity and its uplift to superspace.  We also study general presymplectic BV-AKSZ models related by adding or removing contractible spacetime dimensions in order to systematically relate the spacetime and superspace formulations of the same system within the AKSZ-like framework. These relations are then illustrated using the supersymmetric particle as a toy model. 

\end{abstract}

\vfill

\blfootnote{${}^{\dagger}$ Supported by the ULYSSE Incentive
Grant for Mobility in Scientific Research [MISU] F.6003.24, F.R.S.-FNRS, Belgium.}

\blfootnote{${}^{\ddagger}$ Also at Lebedev Physical Institute Moscow, Russia}

\newpage

\tableofcontents

\section{Introduction}

Batalin-Vilkovisky formalism~\cite{Batalin:1981jr,Batalin:1983wj} gives a very general and systematic framework for gauge theories both at classical and quantum level. The price to be paid for the universality of this approach is that in applications to concrete local gauge theories the geometrical structures underlying gauge transformations, equations of motion, global symmetries, and so on are somewhat hidden in the BV field-antifield space involving a bunch of auxiliary variables such as ghosts, ghost for ghosts, and their conjugated antifields.

In the case of topological theories the BV formulation can be made very concise and geometrical by employing the celebrated  AKSZ~\cite{Alexandrov:1995kv} construction. In so doing the BV master action and symplectic structure are encoded in the geometrical structure of the typically finite-dimensional target space while  physical fields, ghost fields, antifields, etc. emerge as different components of the inhomogeneous differential forms on the spacetime manifold. Moreover, gauge transformations, equations of motion, etc. acquire  a clear geometrical meaning in terms of the space of maps from the source to the target.

Various extensions of the AKSZ construction to nontopological gauge theories are known by now. One possibility is to allow for infinite-dimensional target space~\cite{Barnich:2006hbb,Barnich:2010sw,Grigoriev:2010ic,Grigoriev:2012xg} directly related to the jet-bundle BV description of the system as well as to to the geometrical approach to PDEs~\cite{Vinogradov1981,Krasil?shchik-Lychagin-Vinogradov,Krasil'shchik:2010ij} and the so-called unfolded approach~\cite{Vasiliev:1988xc,Vasiliev:2005zu} of higher spin theory. Another possibility is to replace the spacetime exterior algebra with a differential graded commutative algebra which is not freely generated~\cite{Bonechi:2009kx,costello2011renormalization,Bonechi:2022aji}, see also~\cite{Costello:2016mgj,Cederwall:2023wxc}.\footnote{See also~\cite{Cattaneo:2023cnt,Arvanitakis:2024dbu,Borsten:2024alh} for somewhat related modifications of the AKSZ construction.}  In this work we mostly employ an alternative extension of the AKSZ construction in which the target-space symplectic structure is allowed to degenerate so that the equations of motion do not force all the components of the generalized curvature to vanish and the system perfectly describes nontopological theories in terms of the finite-dimensional target space~\cite{Alkalaev:2013hta,Grigoriev:2016wmk,Grigoriev:2020xec}, see also~\cite{Sharapov:2021drr,Dneprov:2022jyn,Grigoriev:2024ecv}. The up-to-date version of the approach is known as the presymplectic BV-AKSZ formalism~\cite{Grigoriev:2022zlq,Dneprov:2024cvt} and gives a very concise description of the BV formulation of local gauge theories in terms of the finite-dimensional geometry of presymplectic almost $Q$-manifolds.  

The crucial feature of the AKSZ construction is that it allows to systematically add/remove contractible  space-time dimensions while keeping the system merely equivalent~\cite{Barnich:2006hbb}, see also ~\cite{Vasiliev:2001zy,Barnich:2006pc,Bonechi:2009kx,Bekaert:2009fg,Alekseev:2010ud,Cattaneo:2012qu} where analogous ideas appeared in the related contexts. For instance, adding a time direction to the AKSZ model describing Hamiltonian BFV formulation, results in the BV formulation of the same theory and the other way around~\cite{Grigoriev:1999qz,Barnich:2003wj}. In the more restrictive setup of free differential algebras the analogous relations are also known for quite some time~\cite{Sullivan:1977fk,DAuria:1982mkx,Castellani:1991et,Ponomarev:2010ji,Misuna:2013ww}. The same idea is employed in constructing formulations where certain symmetries are realised in a geometrical way. For instance, analogous relations are employed in deriving manifestly AdS or conformal invariant formulations of gauge fields~\cite{Barnich:2006pc,Bekaert:2009fg,Alkalaev:2011zv,Bekaert:2013zya} starting from their formulations in terms of the ambient space where the AdS/conformal symmetry acts linearly and geometrically. 

This feature becomes especially useful in the context of supersymmetric systems where the formulations in terms of the superspace, where supersymmetry acts geometrically, are particularly relevant. For instance, the superspace formulation of the simplest supermultiplet can be derived systematically, starting from its formulation as a free differential algebra in Minkowski space~\cite{Ponomarev:2010ji,Misuna:2013ww}. However, the direct application of the Lagrangian AKSZ construction to superspace formulations is not so straightforward as differential forms on superspaces can not be integrated without extra structures, see~\cite{Grassi:2016apf,Salnikov:2016bny,Hulik:2022hpb} for the discussion of the superspace AKSZ sigma-models.

In this work we elaborate on the presymplectic BV-AKSZ approach to supersymmetric systems. In particular, we construct a presymplectic BV-AKSZ formulation of the minimal $N=1$ supergravity in 4 dimensions~\cite{Freedman:1976xh}, giving a concise and 1st order version of the known BV formulation~\cite{Baulieu:1990uv}, see also~\cite{Brandt:1996au,Brandt:2002rs,Boulanger:2018fei}.\footnote{
After the first version of this preprint appeared, the explicit derivation of the  BV extension of the 1st order $N=1$, $D=4$ supergravity action has been constructed \cite{Cattaneo:2025wbz} in the usual way using Darboux coordinates for the BV symplectic structure. As far as the BV formulation of other supergravity models are concerned let us also mention recent papers~\cite{Cattaneo:2024sfd,Kupka:2025hln} which appeared when this work was being completed.} Because any presymplectic BV-AKSZ system can be considered at the level of equations of motion, it can be naturally uplifted to the extended spacetime. The presymplectic AKSZ formulation of $N=1$ $D=4$ supergravity lifted to the respective ``group manifold'' naturally reveals the main objects of the so-called group manifold approach to (super)gravity put forward in~\cite{DAuria:1982mkx,Castellani:1991eu}, see also \cite{Fre:2008qw,Castellani:2018zey} for a  more recent discussion. More precisely, the extended system describes a version of the corresponding Cartan connection on the total space of the underlying fiber bundle.

While the uplift of supergravity to the superspace is done at the level of equations of motion only, we demonstrate that the analogous treatment of the toy-model of supersymmetric particle can be easily performed at the Lagrangian level as well. To this end we introduce the extension of the presymplectic BV-AKSZ formalism, where the source space is replaced with a suitable differential graded commutative algebra with the trace, giving a concise framework to study the relations between presymplectic BV-AKSZ models obtained by adding/removing space-time directions. This can be also considered as a presymplectic generalization of the extension of the AKSZ construction put forward in~\cite{Bonechi:2009kx,costello2011renormalization,Bonechi:2022aji}. It turns out that employing suitable degenerate traces, the supersymmetric models can be cast into the presympletic AKSZ form. In our toy model example, the necessary modification of the trace can be understood as the insertion of a suitable picture-changing operator from~\cite{Castellani:2017ycm}. We explicitly illustrate all the steps using the toy-model of supersymmetric particle. In particular, we demonstrate that the picture-changing operator employed in the BV-AKSZ formulation of the model on the superspace, defines a presymplectic BV-AKSZ formulation.

The paper is organized as follows: Section~\bref{sec:prelim} is devoted to the background material. This includes presymplectic AKSZ construction illustrated by the example of the Palatini-Cartan-Weyl (frame-like) formulation of Einstein gravity. There we also review the conventional Lagrangian formulation of the gravitino field and the full $N=1$, $D=4$ supergravity. In section~\bref{sec:sugra} we present a presymplectic BV-AKSZ formulation of $N=1$ $D=4$ supergravity and show that it reproduces both the classical Larangian and its BV extension in a very concise form. General framework to study (presymplectic) BV-AKSZ formulations on spacetimes related by adding/removing contractible dimensions is given in Section~\bref{sec:contract}. Then in Section \bref{sec:rheonomy} we study the uplift of (super)gravity to the superspace and show that it naturally reproduces the respective Cartan geometry. The relation between the (presymplectic) BV-AKSZ formulations in the spacetime and the superspace are studied in details in Section~\bref{sec:particle} using the example of $N=1$ supersymmetric particle. In particular, we systematically derive the respective superfield action by uplifting the spacetime BV-AKSZ formulation to the superspace. Some technical details and a summary of conventions are given in the Appendices.

\section{Preliminaries}
\label{sec:prelim}
\subsection{Massless spin $3/2$ in Minkowski space}
Massless spin-$3/2$ is described by a spacetime  1-form with values in Majorana spinors $\psi_\mu$ ($\overline{\psi_\mu} = \psi_\mu^\dagger \gamma_0 = \psi_\mu^T C$). We will use notation $\psi = \psi_\mu dx^\mu$ and assume that $\psi_\mu$ anticommute, i.e. $\epsilon(\psi_\mu) =\p{\psi_\mu}=1$, see Appendix~\bref{sec:app-conventions} where our conventions are summarized. The Lagrangian is given by 
\begin{equation}
    L = \epsilon^{\mu\nu\rho\sigma} \overline{\psi}_\mu \gamma_5 \gamma_\nu \partial_\rho \psi_\sigma.
\end{equation}
and its equations of motion are the massless Rarita–Schwinger equations:
\begin{equation}
    \epsilon^{\mu\nu\rho\sigma} \gamma_5 \gamma_\nu \partial_\rho \psi_\sigma = 0.
\end{equation}

The Lagrangian is invariant under the following gauge transformations:
\begin{equation}
    \psi_\mu \to \psi_\mu + \partial_\mu \varepsilon, 
\end{equation}
where the Majorana spinor $\varepsilon$ is a gauge parameter. Together with the standard gauge-fixing condition $\gamma^\mu \psi_\mu = 0$ the equations of motion give the full set of irreducibility conditions:
\begin{equation}
    \partial^\mu \psi_\mu = 0, \qquad \gamma^\nu \partial_\nu \psi_\mu = 0, \qquad \Box \psi_\mu = 0.
\end{equation}

\subsection{$N=1$, $D=4$ Supergravity}

The action of the minimal $N=1$ supergravity in four spacetime dimensions is that of spin $3/2$ coupled to Palatini-Cartan-Weyl gravity and can be written as:
\begin{equation} \label{N=1 D=4 sugra action}
    S =  \int \half R^{ab} \wedge e^c \wedge e^d \epsilon_{abcd} - 2i \overbar{\psi} \wedge \gamma_5 \gamma_a d_\Gamma \psi \wedge e^a, \quad d_\Gamma \psi = dx^\mu D_\mu \psi = dx^\mu (\partial_\mu + \frac{1}{4}\omega_\mu^{ab}\gamma_{ab})\psi,
\end{equation}
where $\gamma_{ab} = \frac{1}{2}[\gamma_a, \gamma_b]$, $R^{ab} = d\omega^{ab} + \omega^a_{\enspace c} \wedge \omega^{cb}$.

Equations of motion have the following form:
\begin{subequations}
\begin{equation}\label{eq-of-motion-omega}
    \delta_\omega S = 0 \, \Rightarrow \, R^a =  T^a - \frac{1}{2} \overbar{\psi}\wedge \gamma^a \psi = 0,
\end{equation}
\begin{equation}\label{eq-of-motion-e}
    \delta_e S = 0 \, \Rightarrow \, R^{ab} \wedge e^c \epsilon_{abcd} - 2i\overbar{\psi} \wedge \gamma_5\gamma_d d_\Gamma \psi = 0,
\end{equation}
\begin{equation}\label{eq-of-motion-psi}
    \delta_\psi S = 0 \, \Rightarrow \, 2 \gamma_a d_\Gamma \psi \wedge e^a + \gamma_a \psi \wedge R^a = 0,
\end{equation}
\end{subequations}
where $T^a = de^a + \omega^a_{\enspace b}e^b$.

Action~\eqref{N=1 D=4 sugra action} is invariant under the following supersymmetry transformations:
\begin{equation}
    \begin{gathered}
    \delta e^a = \overbar{\epsilon}\gamma^a \psi,\\    
    \delta \psi = d_\Gamma \epsilon.
\end{gathered}
\end{equation}
provided algebraic constraint $R^a = 0$ holds. The 1st order action \eqref{N=1 D=4 sugra action} is still invariant under local supersymmetry but the above transformations are to be extended to Lorentz connection. Note that the the algebra of gauge transformations in the 1st order formulation  does not close off-shell. 


\subsection{Presymplectic BV-AKSZ construction}

In the context of BV-BRST approach to systems involving physical fermions one has to take care of two degrees: the fermion degree denoted by $\epsilon(\cdot)$ and the ghost degree denoted by $\gh{\cdot}$. The total Grassmann degree which determines the graded commutativity is denoted by $\p{\cdot}$ and is given by $\p{A}=\gh{A}+\epsilon{(A)}\,\, \mathrm{mod}\,\,2$. For instance, for homogeneous functions $A,B$ on the graded supermanifold one has 
$AB=(-1)^{\p{A}\p{B}}BA$. Unless otherwise specified we assume that the Grassmann degree of the main structures such as BV differential or symplectic structure, are compatible with the ghost degree, e.g.~$\p{Q}=\gh{Q} \,\, \mathrm{mod}\,2$.  

Let $(F,\omega^F,Q)$ be a weak presymplectic $Q$-manifold~\cite{Grigoriev:2022zlq}, i.e. a graded supermanifold equipped with vector field $Q,\p{Q}=\gh{Q}=1$ and presymplectic form $\omega^F = d\chi, \gh{\omega^F}=k$ such that,
\begin{equation}
 L_Q \omega^F = 0\,,\qquad  i_Qi_Q( \omega^F)=0\,. 
\end{equation}
Note that these relation imply $i_{Q^2}\omega^F=0$ so that $Q$ is nilpotent on the symplectic quotient. It follows there exists Hamiltonian $H\in C^\infty(F)$ such that\footnote{For $\gh{\omega^F}={-1}$ the Hamiltonian may exist only locally.}
\begin{equation}
    i_Q \omega^F+dH=0\,.
\end{equation}
Now let us consider a real manifold $X$ (space-time) of dimension $\dim(X) = \gh{\omega^F}-1$ and its shifted tangent bundle $T[1]X$ with coordinates $(x^\mu, \atheta^\nu)$. $T[1]X$ is naturally a $Q$-manifold with the $Q$-structure being the de Rham differential: $\dx = \atheta^\mu \frac{\partial}{\partial x^\mu}$. 

Fields of the presymplectic BV-AKSZ sigma model $((F,Q,\omega^F),(T[1]X,\dx))$   are the 
ghost-degree preserving supermaps $\sigma: T[1]X \to F$. The (presymplectic) AKSZ action is given by 
\begin{equation}
\label{pAKSZ-action}
    S[\sigma] = \int_{T[1]X}  \left[ \sigma^\ast(\chi)(d_X) + \sigma^\ast(H)\right]\,,
\end{equation}
and its Euler-Lagrange equations take the form
\begin{equation}\label{pAKSZ-eq}
    \sigma^\ast(\omega^F_{AB})\left(d_X\sigma^\ast(\psi^B) - \sigma^\ast(Q^B)\right) = 0\,,
\end{equation}
where $\psi^A$ are coordinates on $F$ and $Q = Q^A\frac{\partial}{\partial \psi^A}, \,\omega^F = \omega^F_{AB}d\psi^Bd\psi^A$. Gauge transformations are:
\begin{equation}\label{pAKSZ-gauge-sym}
    \delta \sigma^\ast = \sigma^\ast \circ [Q+\dx,Y],
\end{equation}
where the gauge parameter $Y$ is a vertical vector field of ghost degree $-1$ on $T[1]X \times F$, which preserves $\omega^F$ pulled back to $T[1]X \times F$. 

Note that if $R=R^A\dl{\psi^A}$ is a regular vector filed on $F$ such that $i_R\omega^F=0$, it generates an algebraic gauge transformation of~\eqref{pAKSZ-action}. In other words, if $\psi^1$ is a coordinate such that $\omega^F_{1A}=0$ then the above action does not depend on the field $\psi^1$.  The BV action extending \eqref{pAKSZ-action} is obtained by considering a generic supermap $\hat\sigma$ and taking a symplectic quotient of the space of supermaps by the kernel distribution of the degree $-1$ presymplectic structure induced by $\omega^F$ on the space of supermaps, see~\cite{Grigoriev:2020xec,Dneprov:2022jyn,Grigoriev:2022zlq,Dneprov:2024cvt} for further details.

The presymplectic BV-AKSZ construction described above admits a straightforward generalisation analogous to that of the usual symplectic AKSZ construction. More precisely in place of $(T[1]X,\dx)$ one can take a more general $Q$-manifold $(\manX,\delta)$ equipped with a $\delta$-invariant density $\rho$. What is interesting is that in the presymplectic setup it is also natural to allow for a possibly degenerate $\rho$ that could lead
to the addition degeneration of the presymplectic structure on the space of supermaps. 
This generalisation turns out to be very useful in the case of superspace AKSZ  models and we discuss it in more details in Section~\bref{sec:generalAKSZ}.


\subsection{Palatini-Cartan-Weyl gravity as a presymplectic BV-AKSZ}
\label{pAKSZ-Gravity}

As an example let us recall the presymplectic BV-AKSZ formulation of Einstein gravity.
Following~\cite{Grigoriev:2020xec,Alkalaev:2013hta}, let $F = \mathfrak{iso}(1,n-1)[1]$ be a shifted Poincar\'e algebra equipped with the Chevalley-Eilenberg differential $Q$ seen as a homological vector field on $F$
\begin{equation}
    Q = -\rho^a_{\enspace b} \xi^b \frac{\partial}{\partial \xi^a} - \rho^a_{\enspace c}\rho^{cb} \frac{\partial}{\partial \rho^{ab}}.
\end{equation}
and the following presymplectic form of ghost degree $n-1$:
\begin{equation}
    \omega^F = \mathcal{V}_{abc}(\xi)d\xi^a \wedge d\rho^{bc} = d(\mathcal{V}_{ab}(\xi)d\rho^{ab}), \quad \mathcal{V}_{a_1\dots a_k}(\xi) = \frac{1}{(n-k)!}\epsilon_{a_1 \dots a_k, b_1 \dots b_{n-k}} \xi^{b_1}\dots \xi^{b_{n-k}}\,.
\end{equation}
This $\omega^F$ is $Q$-invariant so that
\begin{equation}
    i_Q \omega^F = -dH = -d(\mathcal{V}_{ab}(\xi) \rho^a_{\enspace c} \rho^{cb}).
\end{equation}

Now let us consider 4-dimensional space-time manifold $X$. Hereafter we assume that $X = \mathbb{R}^n$  for simplicity. Maps $\sigma: T[1]X \to F$ can be parameterized by
\begin{equation}
    \sigma^\ast (\rho^{ab}) = \omega^{ab}_{\mu}(x)\atheta^\mu, \quad \sigma^\ast (\xi^a) = e^a_{\mu}(x)\atheta^\mu.
\end{equation}
Identifying functions on $T[1]X$ with differential forms on $X$ the AKSZ-action takes the following form:
\begin{equation}
    S_{GR}[e,\omega] = \int_X \frac{1}{(n-2)!}\epsilon_{a_1 \dots a_{n-2}bc} e^{a_1} \wedge \dots e^{a_{n-2}}R^{bc}\,,\qquad R^{ab}=d\omega^{ab}+\omega^a_{\enspace c}\omega^{cb}\,.
\end{equation}
This is the well-known Palatini-Cartan-Weyl action of gravity. The complete BV formulation is obtained by taking the symplectic quotient of the space of the supermaps.


\section{Presymplectic  BV-AKSZ formulation of $N=1$, $D=4$ Supergravity}
\label{sec:sugra}

\subsection{Basic structures}

Let $\algg=\mathfrak{siso(1,3)}$ be the $N=1$ Super-Poincar\'e algebra in four dimensions. The commutation relation in the standard basis can be written as: 
\begin{equation}
\begin{gathered}
    [M_{ab}, M_{cd}] = \eta_{bc} M_{ad} + \eta_{ad}M_{bc} - \eta_{bd}M_{ac} - \eta_{ac}M_{bd},\\
    [M_{ab}, P_c] = \eta_{bc} P_a - \eta_{ca} P_b,\\
    \{\mathcal{Q}_\alpha, \mathcal{Q}_\beta\} =- (\gamma^a)_{\alpha\beta} P_a,\\
    [M_{ab}, \mathcal{Q}_\beta] = \frac{1}{2} \mathcal{Q}_\alpha (\gamma_{ab})^\alpha_{\enspace \beta}\,.
\end{gathered}
\end{equation}
Our notations are summarized in Appendix~\bref{sec:app-conventions}. If we denote by $e_I$ all the basis elements of $\algg$, the coordinates on the associated graded supermanifold $\algg[1]$ are $\Psi^I$ such that $\gh{\Psi^I}=1$ and $\p{\Psi^I}=1-\p{e_I}~mod~2$,
where $\p{\cdot}$ denotes the Grassmann parity of an element of $\algg$ or a function on $\algg[1]$. In more details:
\begin{equation}
    \begin{gathered}
    \rho^{ab},  \quad \gh{\rho^{ab}} = 1, \quad \p{\rho^{ab}} = 1,\\
    \xi^a, \quad  \gh{\xi^a}=1 \quad \p{\xi^a} = 1,\\
    \psi^\alpha, \quad \gh{\psi^\alpha} = 1, \quad  \p{\psi^\alpha} = 0.
\end{gathered}
\end{equation}
where $\rho^{ab}$ is associated to $M_{ab}$, $\xi^a$ to $P_a$, and $\psi^\alpha$ to $\mathcal{Q}_\alpha$. 

The Chevalley-Eilenberg differential, understood as a homological vector field on $\algg[1]$ is defined in a standard way via $Q\mathbf{\Psi}=-\half\commut{\mathbf{\Psi}}{\mathbf{\Psi}}$, where $\mathbf{\Psi}=\Psi^I \tensor e_I$ is a distinguished element of 
$\cC^\infty(\algg[1])\tensor \algg$. Note that $\mathbf{\Psi}$ is an analog of "string field" known in the BRST first-quantized approach.  In terms of the components one has:
\begin{equation}
    Q = -\left(\rho^a_{\enspace b} \xi^b - \frac{1}{2}(\gamma^a)_{\alpha\beta} \psi^\alpha \psi^\beta \right) \frac{\partial}{\partial \xi^a} - \rho^a_{\enspace c} \rho^{cb} \frac{\partial}{\partial \rho^{ab}} - \frac{1}{4} (\gamma_{ab})^\alpha_{\enspace \beta} \rho^{ab}\psi^\beta \frac{\partial}{\partial \psi^\alpha}.
\end{equation}

The next ingredient we need is a compatible presymplectic form of degree $3$. Consider the following $2$-form
\begin{equation}
    \begin{gathered}
    \label{omega-F}
    \omega^F =  \epsilon_{abcd} \xi^d d\xi^a d\rho^{bc} - 2i (\gamma_5 \gamma_a)_{\alpha\beta}  \xi^a d \psi^\alpha  d \psi^\beta  + 2i (\gamma_5 \gamma_a)_{\alpha\beta} \psi^\alpha d\psi^\beta d\xi^a = d\chi,\\
    \chi = \frac{1}{2}\epsilon_{abcd} \xi^c \xi^d d\rho^{ab} - 2i (\gamma_5 \gamma_a)_{\alpha\beta} \psi^\alpha d\psi^\beta \xi^a.
    \end{gathered}
\end{equation}
One can show $\omega^F$ is indeed $Q$-invariant or equivalently:
\begin{equation}
    i_Q \omega^F =- dH= -d \left(\frac{1}{2}\epsilon_{abcd}\xi^c \xi^d \rho^a_{\enspace e} \rho^{eb} - \frac{i}{2} \psi^\alpha ( \gamma_5 \gamma_a \gamma_{bc})_{\alpha\beta} \rho^{bc} \psi^\beta \xi^a\right).
\end{equation}
The proof is relegated to Appendix~\bref{app:iq=dH}.

Now let $X$ be a 4 dimensional space-time manifold. The ghost degree preserving supermaps $\sigma: T[1]X \to \algg[1]$ can be parameterized as follows:
\begin{equation}\label{Sugra map space}
    \begin{gathered}
    \sigma^\ast(\xi^a) = e^a = e^a_\mu (x)\atheta^\mu, \quad \sigma^\ast(\rho^{ab}) = \omega^{ab} = \omega^{ab}_\mu(x) \atheta^\mu, \quad \sigma^\ast(\psi^\alpha) = \psi^\alpha = \psi^\alpha_{\enspace \mu}(x)\atheta^\mu,\\
    \p{e^a_\mu} = 0, \qquad \p{\omega^{ab}_\mu} = 0, \qquad \p{\psi^\alpha_{\enspace \mu}} =1.
\end{gathered}
\end{equation}
The AKSZ-like action~\eqref{pAKSZ-action} takes the following form:
\begin{equation}
    S_{sugra}[e,\omega,\psi] = \int_X\left( \frac{1}{2} \epsilon_{abcd} (d\omega^{ab} + \omega^a_{\enspace e} \wedge \omega^{eb})\wedge e^c \wedge e^d - 2i \overbar{\psi} \wedge \gamma_5 \gamma_a (d\psi + \frac{1}{4}\omega^{bc}\gamma_{bc}\wedge \psi) \wedge e^a\right) d^4x
\end{equation}
and indeed coincides with \eqref{N=1 D=4 sugra action} up to an overall factor.

For completeness, let us also write down the explicit expressions for $\omega^F, Q$, and $H$ in the two-component spinor notations for 4 dimensional spinors and tensors.  More precisely we employ the standard conventions from~\cite{penrose1984spinors,huggett1994introduction}. The expressions for    
the main target space structures take the form:
\begin{equation}\label{SUGRA presymplectic structure}
   \omega^F = \frac{i}{2}\left(\xi_A^\gammadot d\xi_{ \B\gammadot}d\rho^{A \B} - \xi_\alphadot^\C d\xi_{\betadot\C}d\rhobar^{\alphadot\betadot}\right)- 4id\psibar^\alphadot d\psi^A \xi_{A\alphadot} -2i (d\psibar^\alphadot \psi^{A} - \psibar^\alphadot d\psi^A) d\xi_{A\alphadot},
\end{equation}
\begin{multline}\label{Q-2c-spinors}
 Q = -\left(\frac{1}{2}(\rho^{A}_{\enspace \B}\xi^{ \B\alphadot} + \rhobar^\alphadot_{\enspace\betadot}\xi^{\betadot A})+2\psibar^\alphadot\psi^A\right)\frac{\partial}{\partial \xi^{A\alphadot}} -\frac{1}{2} \left(\rho^A_{\enspace\C}\rho^{\C \B}\frac{\partial}{\partial \rho^{A \B}} +\rhobar^\alphadot_{\enspace\gammadot}\rhobar^{\gammadot\betadot}\frac{\partial}{\partial \rhobar^{\alphadot\betadot}}\right) - \\
    -\frac{1}{2}\left(\rho^A_{\enspace \B}\psi^ \B \frac{\partial}{\partial \psi^A}  + \rhobar^\alphadot_{\enspace\betadot}\psibar^\betadot\frac{\partial}{\partial \psibar^\alphadot}\right),
\end{multline}
\begin{equation}
    H = \frac{i}{4}(\xibar_{\alphadot\betadot} \rhobar^{\alphadot}_{\enspace\gammadot} \rhobar^{\gammadot\betadot}-\xi_{A \B}\rho^{A}_{\enspace\C} \rho^{\C \B}) -i \psibar^\alphadot\psi^A(\rho_{A}^{\enspace \B}\xi_{\B\alphadot} - \rhobar_{\alphadot}^{\enspace\betadot}\xi_{\betadot A}).
\end{equation}

The presymplectic BV-AKSZ formulation proposed in this work applies, strictly speaking,  to the case where the tangent bundle to $X$ is globally trivial. However, it can be directly extended to generic spacetime by constructing a $Q$-bundle over $T[1]X$, which is locally isomorphic to the present system. The resulting formulation is of more general type than presymplectic BV-AKSZ model and is known as presymplectic gauge PDE, see~\cite{Grigoriev:2022zlq,Dneprov:2024cvt}.

\subsection{Regularity of the presymplectic structure}
To make sure that the above construction gives not only the Lagrangian but also the full-scale BV formulation one needs to check that the presymplectic structure induced on the space of supermaps is regular provided one restricts to admissible supermaps which are  supermaps such that their frame field $e_\mu^a$ component is invertible. Another requirement is that the resulting BV master-action is proper, i.e. it takes all the gauge symmetries into account.

To study regularity it is convenient to introduce a graded supermanifold $\bar F=Smaps(T_x[1]X,F)$ of supermaps from $T_x[1]X$ to $F$, where $x\in X$ is a generic point of $X$. In contrast to $Smaps(T[1]X,F)$, $\bar F$ is finite-dimensional and it is enough to check that the presymplectic structure induced on $\bar F$ by $\omega^F$ on $F$ is regular.

Supermaps can be seen as homomorphsims of the respective algebras of functions $\hat\sigma^*:\cC^\infty(F) \to \cC^\infty(\bar F \times T_x[1]X)$. The coordinates on $\bar F$ are introduced  as follows:
\begin{equation}
\label{coord-barF}
\begin{gathered}
    \hat\sigma^\ast(\xi^a) = \xi^a + e^a_\mu \atheta^\mu + \frac{1}{2!}\st{2}{\xi}{}^a_{\mu\nu} \atheta^\mu \atheta^\nu+\frac{1}{3!}\st{3}{\xi}{}^a_{\mu\nu\lambda} \atheta^\mu\atheta^\nu\atheta^\lambda+\frac{1}{4!}\st{4}{\xi}{}^a_{\mu\nu\lambda\rho} \atheta^\mu\atheta^\nu\atheta^\lambda\atheta^\rho,\\
    \hat\sigma^\ast(\rho^{ab}) = \rho^{ab} + \omega^{ab}_\mu \atheta^\mu + \frac{1}{2!}\st{2}{\rho}{}^{ab}_{\mu\nu} \atheta^\mu \atheta^\nu+\frac{1}{3!}\st{3}{\rho}{}^{ab}_{\mu\nu\lambda} \atheta^\mu\atheta^\nu\atheta^\lambda+\frac{1}{4!}\st{4}{\rho}{}^{ab}_{\mu\nu\lambda\rho} \atheta^\mu\atheta^\nu\atheta^\lambda\atheta^\rho,\\
    \hat\sigma^\ast(\psi^\alpha) = \psi^\alpha + \psi^\alpha_\mu \atheta^\mu + \frac{1}{2!}\st{2}{\psi}{}^\alpha_{\mu\nu} \atheta^\mu \atheta^\nu+\frac{1}{3!}\st{3}{\psi}{}^\alpha_{\mu\nu\lambda} \atheta^\mu\atheta^\nu\atheta^\lambda+\frac{1}{4!}\st{4}{\psi}{}^\alpha_{\mu\nu\lambda\rho} \atheta^\mu\atheta^\nu\atheta^\lambda\atheta^\rho\,.
\end{gathered}
\end{equation}
In this coordinates the pullback of the evaluation map $ev:\bar F \times T_x[1]X \to F$ 
is given by the same equations \eqref{coord-barF} with $\hat\sigma^*$ replaced by $ev^*$. The presymplectic structure induced on $\bar F$ can be written as
\begin{equation}
\bar\omega=\int_{T_x[1]X} ev^*(\omega^F) \,,
\end{equation}
and by construction is a 2-form on $\bar F$. Strictly speaking, linear change of coordinates on $T_x[1]X$ results in the Jacobian factor so that the form $\omega_{BV}$ that $\bar\omega$ induces on $Smaps(X,\bar F)\simeq Smaps(T[1]X,\bar F)$ is already a genuine BV presymplectic form. However this subtlety is irrelevant for the analysis of regularity. For completeness, let us also give an expression for the BV master action in terms of the evaluation by extending the evaluation map from $T_x[1]X$ to the entire $T[1]X$ so that $ev: \bar \cF\times T[1]X \to F$, $\bar\cF$ is the space of supermaps from $T[1]X$ to $F$. Then the BV action seen as a function on $\bar\cF$ can be written as
\begin{equation}
S_{BV}=\int_{T[1]X} (ev^*(\chi)(\dx)+ev^*(H))\,.
\end{equation}
Note that $\bar\cF$ is a presymplectic manifold so that the standard BV formulation is obtained by taking a symplectic quotient. More details on supermaps and the associated evaluation map can be found in e.g.~\cite{Roytenberg:2006qz}.

We also need to recall how vector fields on $F$ are naturally prolonged to $\bar F$.  Let $V$ be a vector field on $F$, its prolongation $\bar V$ is a vector field on $\bar F$ determined by
\begin{equation}
\bar V ev^*(\Psi^I)=ev^*(V\Psi^I)\,,
\end{equation}
where $\Psi^I$ denote all coordinates on $F$.

First we verify that the correct spectrum of BV fields and antifields is indeed recovered.  We consider the presymplectic structure at  generic point  $p\in \mathrm{body}(\bar F)$ of the body of $\bar F$, i.e. a point where all the coordinates of nonvanishining ghost degree or Grassmann parity vanish. These are all the coordinates introduced in \eqref{coord-barF} save for $e^a_\mu,\omega_\mu^{ab}$. It is also convenient to adjust the  basis in such a way that $e^a_\mu=\delta^a_\mu$. The explicit expression for the presymplectic structure
at $p$ reads as
\begin{multline}\label{Sugra-presymplectic-form-body}
    \bar\omega_p = \bar\omega_p^{PCW} + \int d^4\atheta \left[-4i(\gamma_5\gamma_a)_{\alpha\beta}\atheta^a \left(\frac{1}{3!}d\psi^\alpha d\st{3}{\psi}{}^\beta_{bcd} - \frac{1}{2!}d\psi_b^\alpha  d\st{2}{\psi}{}_{cd} \right)\atheta^b \atheta^c\atheta^d\right]=\\
   = \bar\omega_p^{PCW} - 4i(\gamma_5\gamma_a)_{\alpha\beta}\left(\frac{1}{3!}d\psi^\alpha d\st{3}{\psi}{}^\beta_{bcd} - \frac{1}{2!}d\psi_b^\alpha d\st{2}{\psi}{}_{cd} \right)\epsilon^{abcd},
\end{multline}
where $\bar\omega_p^{PCW}$ is the presymplectic structure in the sector of $\xi^a$ and $\rho^{ab}$ which is known to produce a correct BV spectrum in the gravity  sector, see~\cite{Grigoriev:2020xec}. In other words, at this point antifields for $\psi^\alpha$ are parameterized by $-4i(\gamma_5\gamma^a)(*\st{3}{\psi})_a$ while antifields for ${\psi}{}^\alpha_b$ by $-4i(\gamma_5\gamma_a)(*\st{2}{\psi})^{ba}$,
where $*$ denotes the Hodge conjugation acting on world indices. It follows that the BV spectrum of fields and antifields is correct in the sector of gravitino field as well.

In order to prove that $\bar\omega$ defined on $\bar F$ is indeed regular, we employ the technique similar to that used in~\cite{Grigoriev:2020xec} to prove the analogous statement in the case of Einstein gravity. Let $\bar\cK$ be the kernel distribution of $\bar \omega$ and denote by $\bar\cK_p \subset T_p\bar F$ the kernel of $\bar\omega_p$ in $T_p \bar F$. In the case of graded supermanifolds, distributions are submodules of vector fields seen as modules over the algebra of functions on the manifold. A distribution is called regular if locally it is freely generated. We have:
\begin{prop}\label{Sugra-presymplectic-form-regularity-distribution-continuation}
    There exists a distribution $\bar\cK'\subset \bar\cK$ on $\bar F$ which is regular and at any point $p \in \mathrm{body}(\bar F)$ coincides with $\bar\cK_p$.
\end{prop}
\begin{proof}
We first check that any vector $\bar k_p\in \bar\cK_p\subset T_p \bar F$ is the restriction to $p$ of a vector field from $\bar\cK$. To this end let us 
decompose $\omega^F$ defined in  \eqref{omega-F} as $\omega^F=\omega_0^F+\omega_1^F$, where $\omega_0^F$ is the contribution proportional to $\xi^a$ (the first two terms) and $\omega_1^F$ is the contribution linear in $\psi^\alpha$ (the third term). It turns out that 
for any $\bar k_p  \in \bar\cK_p$ it is easy to find a vector field $k$  on $F$ such that $k$ lies in the kernel of $\omega^F_0$ and its prolongation to $\bar F$ coincides with $\bar k_p$ at $p$. To see this it is enough to consider vectors in the sector of $\xi,\rho$ and in the sector of $\psi$ independently. In the sector of $\xi,\rho$ this follows from the analysis of~\cite{Grigoriev:2020xec}. If a vector $\bar k_p =T^\alpha_{a_1\ldots a_k}\dl{\psi_{a_1\ldots a_k}^\alpha}\in T_p\bar F$ is in the kernel of $\bar \omega_p$, consider the following vector field on $F$:
\begin{equation}
k=T^\alpha_{a_1\ldots a_k}\xi^{a_1}\ldots \xi^{a_k}\dl{\psi^\alpha}\,.
\end{equation}
It is easy to check that it belongs to the kernel distribution of $\omega^F_0$ and the only term in its prolongation that does not vanish at $p$ is precisely $\bar k_p$.  We refer to the prolongation of $k$ as to its extension from $p$ to $\bar F$. Moreover, if two vectors in $T_p \bar F$ are linearly independent then their extensions are also linearly independent over $\cC^\infty(\bar F)$. Note also that the prolongation of $k$ restricted to any $p^\prime\in \textrm{body}\bar F$ is also a nonvanishing vector from the kernel of $\bar\omega_{p^\prime}$.
\begin{lemma} \label{lemma:kernel}
Let $k$ be a vector field on $F$ that belongs to the kernel of $\omega^F_0$. Then there exists a vector field $k^\prime$ such that it is proportional to $\psi^\alpha$ and  $k+k^\prime$ belongs to the kernel of $\omega^F=\omega^F_0+\omega^F_1$.
\end{lemma}
The proof is straightforward but lengthy and requires explicit coordinate considerations so that it is given in the Appendix~\bref{app:kernel of presympelctic form}. 

Because the prolongation of a vector field proportional to $\psi$ vanishes at $p$ this gives the desired result. Namely, as distribution $\bar\cK^\prime$ one takes the distribution generated by the prolongations of the vector fields in the kernel of $\omega^F$ whose restrictions to $p$ are linearly independent. Note that by construction the restriction to $\bar\cK^\prime$ to $p$ coincides with $\bar\cK_p$. Moreover, because one can generate $\bar \cK^\prime$ by
vector fields extending the basis in $\bar\cK_p$, distribution $\bar \cK^\prime$ is freely generated. 
\end{proof}
The regularity then follow from the following general statement:
\begin{prop}
Let $(M,\omega)$ be a presymplectic manifold, $p\in M$  a point, and $\cK^\prime \subset \cK$ be a subdistribution of the kernel distribution $\cK$ of $\omega$. If $\cK^\prime$ is regular and $\cK^\prime|_p$ coincides with the kernel of $\omega|_p$ in $T_pM$ then there exists a neighborhood of $p$ in $M$ such that $\cK=\cK^\prime$. 
\end{prop}
\begin{proof}
In the case of $M$ real this immediately follows from the fact that the rank of $\omega$ can not drop when moving off the point. If the neighborhood is formal (as in the case at hand) then
the statement is easily seen using local coordinates. Strictly speaking our case is mixed because the body of $\bar F$ is a real manifold, but the rank is constant on the body  of $\bar F$ so only extension in coordinates of nonvanishing degree matters. More precisely, one can perform the analysis for $T\bar F$
pulled back to the surface $e^a_\mu=\const$, $\omega_\mu^{ab}=const$, where the dependence on the base coordinates is formal.  
\end{proof}

\subsection{Generalization to AdS supergravity}

The formulation of the previous Section has a straightforward generalization to the AdS version of the minimal $N=1,D=4$ sugra. Recall, that the presymplectic BV-AKSZ formulation of Einstein gravity with nonvanishing cosmological constant is almost identical to that with the vanishing one~\cite{Grigoriev:2020xec,Alkalaev:2013hta}. The only difference is that the target space $Q$-structure is deformed and corresponds to (A)dS algebra seen as a deformation of the Poincar\'e one. At the same time the presymplectic structure remains unchanged.

It turns out that the situation with the minimal $N=1,D=4$ sugra is completely analogous. The only subtlety is that only negative cosmological constants are allowed because of the propertied of spinors in 4 dimensions. More precisely, the respective Lie algebra is $\mathfrak{osp(4|1)}$ which is a smooth deformation of $\mathfrak{siso}(1,3)$. In terms of the two-component spinor conventions the commutation relations of 
$\mathfrak{osp}(4|1)$ read as
\begin{equation}
\begin{gathered}
    [M_{AB}, P_{CC'}] = \epsilon_{C(A} P_{B)C'},\quad  [M_{AB},M_{CD}] = -2\epsilon_{(A(C}M_{B)D)}, \quad [M_{AB}, \Q_C] =\epsilon_{C(A} \Q_{B)}, \\
     [P_{AA'}, P_{BB'}] = \frac{\lambda}{2}\left(M_{AB}\epsilon_{A'B'} + \overbar{M}_{A'B'} \epsilon_{AB}\right), \quad [P_{AA'}, \Q_B] = -\frac{1}{2}g\Qbar_{A'}\epsilon_{AB}\\
     \{\mathcal{Q}_A,  \overbar{\mathcal{Q}}_{A'}\} = -2P_{AA'}, \quad  \{\Q_A,\Q_B\} = -2gM_{AB}\,,
\end{gathered}
\end{equation}
where $\lambda = \frac{\Lambda}{3}$ is the rescaled cosmological constant,
$|g|^2 = -\lambda$, and the remaining nontrivial commutators are obtained by complex conjugation. Chevalley-Eilenberg differential for $\mathfrak{osp(4|1)}$ reads as
\begin{equation}
\begin{gathered}
    Q\xi^{AA'} = -\frac{1}{2}(\rho^{A}_{\enspace B}\xi^{B\Aprime} + \rhobar^{\Aprime}_{\enspace\Bprime}\xi^{\Bprime A})+2\psibar^\Aprime\psi^A,\\
    Q\rho^{AB} = -\frac{1}{2} \rho^A_{\enspace C}\rho^{C B} + \frac{\lambda}{2}\xi^A_{\enspace A'}\xi^{A'B} +2g\psi^A\psi^B,\\
    Q\psi^A = -\frac{1}{2}\rho^A_{\enspace B}\psi^B - \frac{\gbar}{2} \xi^A_{\enspace A'}\psibar^\Aprime.
\end{gathered}
\end{equation}
It is clear that at $\lambda=0$ it reduces to~\eqref{Q-2c-spinors} corresponding to $\mathfrak{siso}(1,3)[1]$ algebra.

The above $Q$ is still compatible with the presymplectic structure~\eqref{SUGRA presymplectic structure}. To see this it is enough to check that the deformation is. The deformed Hamiltonian reads as:
\begin{multline}
    H = \frac{i}{4}(\xibar_{A'B'} \rhobar^{A'}_{\enspace C'} \rhobar^{C'B'}-\xi_{A B}\rho^{A}_{\enspace C} \rho^{C B}) -i \psibar^\Aprime\psi^A(\rho_{A}^{\enspace B}\xi_{BA'} - \rhobar_{A'}^{\enspace B'}\xi_{B' A}) +\\ 
    +2ig\xi_{AB}\psi^A\psi^B - 2i\gbar\xibar_{A'B'}\psibar^\Aprime\psibar^\Bprime+6\lambda\xitilde\,.
\end{multline}

Now taking $T[1]X$ as the source space, where $X$ is a 4-dimensional manifold, and considering $\mathbb{Z}$-degree preserving supermaps  the corresponding AKSZ action reads as
\begin{multline}
    S[e,\omega,\psi] = \int_X\left( \frac{1}{2} \epsilon_{abcd} (d\omega^{ab} + \omega^a_{\enspace e} \wedge \omega^{eb})\wedge e^c \wedge e^d - 2i \overbar{\psi} \wedge \gamma_5 \gamma_a (d\psi + \frac{1}{4}\omega^{bc}\gamma_{bc}\wedge \psi) \wedge e^a \right.+\\
    \left.+\frac{ig}{2}\epsilon_{abcd}\xi^a\xi^b\psibar\gamma^{cd}\psi+2\Lambda\epsilon_{abcd}e^ae^be^ce^d\right) d^4x\,,
\end{multline}
where we have switched to world indices and have assumed $g$ to be real. This coincides with the known action functional for the minimal AdS supergravity in 4 dimensions~\cite{Castellani:1991eu}.  Because the presymplectic structure is unchanged, all the arguments ensuring that the construction gives a full-scale BV formulation remain valid and we conclude that the model indeed gives a proper presymplectic BV-AKSZ formulation of the minimal $N=1,D=4$ AdS supergravity.


\section{Contractible space-time dimensions}
\label{sec:contract}

Before discussing the relation between the presymplectic BV-AKSZ formulation of supergravity 
and the group manifold approach~\cite{DAuria:1982mkx,Castellani:1991eu} let us make some general observations on how the source space of an AKSZ sigma model can be equivalently extended or reduced. 

\subsection{AKSZ at the level of equations of motion}
\label{sec:source-ext-eom}

Quite often it is convenient to formulate a gauge theory in terms of a certain extension of its space-time manifold because this helps to realise symmetries in a manifest and geometrical way. The well-known examples of this is the ambient-space formulation of (A)dS and conformal fields~\cite{Dirac:1935zz,Dirac:1936fq}, see also~\cite{Bengtsson:1990un,Buchbinder:2001bs,Sagnotti:2003qa,Bonelli:2003zu,Barnich:2006pc,Bekaert:2009fg}. More relevant in the present work are the superspace formulations of supersymmetric systems, where the supersymmery transformation of fields are induced by the super-diffeomophisms of the spacetime manifold, see e.g.~\cite{Castellani:1991eu,Wess:1992cp,Buchbinder:1998qv,Freedman:2012zz}.

A key observation which explains why AKSZ-like models are relevant in relating formulations of the same system in different spacetimes is that adding/removing a space-time dimension in such a way that the system remains essentially equivalent, is a natural geometrical operation~\cite{Barnich:2006hbb}, see also \cite{Grigoriev:1999qz,Vasiliev:2001zy,Vasiliev:2003ar,Barnich:2003wj} for earlier related considerations and
\cite{Bekaert:2009fg,Bonechi:2009kx,Cattaneo:2012qu,Bekaert:2012vt} for further developments. Let us briefly recall how this works at the level of equations of motion.

An AKSZ model at the level of equations of motion (nonlagrangian  AKSZ in what follows) is a pair: the source space $(T[1]X,\dx)$ and the target $(F,Q)$.  In general, one can consider sources which are not of the form $(T[1]X,\dx)$ but we postpone this discussion for Section~\bref{sec:generalAKSZ}. This data defines a local BV system whose space of field configurations is  $Smaps(T[1]X,F)$, i.e. supermaps from $T[1]X$ to $F$. The BV differential therein is a natural lift of $\dx+Q$ on $T[1]X\times F$ to the space of all (including those not preserving the ghost degree) supermaps.

Given a nonlagrangin AKSZ model one can replace its spacetime  $X$ with another manifold. It turns out that adding a contractible dimension, i.e. replacing $X$ with $X \times \fR^1$ so that the new source is $(T[1]X\times T[1]\fR^1,\dx+\dT)$, gives a new AKSZ model which is in a certain sense equivalent to the initial one. This new AKSZ model can be again considered as a gauge field theory on $T[1]X$ by identifying its field-antifield space $Smaps(T[1]X\times T[1]\fR^1,F)$ as
\begin{equation}
Smaps(T[1]X\times T[1]\fR^1,F)\simeq Smaps(T[1]X, Smaps( T[1]\fR^1,F))\,.
\end{equation}
If we allow ourselves to consider AKSZ sigma models with infinite-dimensional target spaces, the above identification defines a new AKSZ sigma model on $T[1]X$ whose target is itself a field-antifield space $Smaps(T[1]\fR^1,F)$ of the  AKSZ sigma model with source $(T[1]\fR^1,\dT)$
and target $(F,Q)$. Moreover, the new AKSZ sigma model on $T[1]X$ is equivalent to the initial one via the elimination of the nonlagrangian version~\cite{Barnich:2004cr,Barnich:2010sw,Grigoriev:2019ojp}  of the generalized auxiliary fields~\cite{Henneaux:1990ua}. Indeed, the initial target space $(F,Q)$ is in fact an equivalent reduction of the new target space $(Smaps( T[1]\fR^1,F),Q^\prime)$, where $Q^\prime$ is the lift of $\dT+Q$ to $(Smaps( T[1]\fR^1,F)$. More precisely, let $t,\theta^\prime$ be coordinates on $T[1]\fR^1$ and $\psi^A$ on $F$. A generic supermap can be written as $\psi_0(t)+\theta^\prime \psi_1(t)$ so that $\psi_0^A(t)$ and $\psi_1^A(t)$ are local coordinates on $Smaps(T[1]\fR^1,F)$. In terms of these coordinates the action of the BV-AKSZ differential $Q^\prime$ on $Smaps( T[1]\fR^1,F)$ reads as:
\begin{equation}
\label{c-pairs}
Q^\prime \psi_0^A(t)=Q^A(\psi_0(t))\,, \qquad Q^\prime \psi_1^A(t)=\dot \psi_0(t)+\ldots
\end{equation}
where $\ldots$ denote terms proportional to $\psi_1^A$. If we split coordinates $\psi^A(t)$ into $\psi_0^A(0)$ and $\dot\psi_0^A(t)$ the above relation imply that the surface in $(Smaps( T[1]\fR^1,F),Q^\prime)$ singled out by $\psi_1^A(t)=0$ and $\dot\psi_0^A(t)=0$ is an equivalent reduction (in the sense of $Q$-manifolds) and can be identified with $(F,Q)$. This in turn implies that the two AKSZ sigma models over $T[1]X$ are equivalent via the elimination of generalized auxiliary fields. In particular, the space of equivalence classes of solutions modulo gauge transformations must be the same. In applications to SUSY systems we also need a version where $t$ is Grassmann odd while $\theta^\prime$ is even. The proof is analogous.

The statement has a straightforward generalization to the case of so called gauge PDEs~\cite{Grigoriev:2019ojp}. Namely, given a gPDE $(E,Q,T[1]X)$ there is a natural gPDE $(E^\prime,Q^\prime,T[1]X\times T[1]\fR^1)$ whose total space is a product $Q$-manifold $(E,Q)\times(T[1]\fR^1,\dT)$ considered as a bundle over $T[1]X\times T[1]\fR^1$. In its turn, it can be seen as a gPDE over $T[1]X$
with typical fiber $Smaps( T[1]\fR^1,F)$, where $F$ is a typical fiber of $E$. A straightforward generalization of the above argument says that $(E,Q,T[1]X)$ can be seen as an equivalent reduction (in the sense of gPDEs, see \cite{Grigoriev:2019ojp}) of the extended one.

\subsection{Presymplectic AKSZ}

The considerations of the previous Subsection can be applied to Lagrangian AKSZ sigma models if we disregard the symplectic structure and consider them at the level of equations of motion. Indeed, for a usual AKSZ sigma model the BV-BRST differential can be defined without using the symplectic structure on the target space.

However, if we are interested in the full-scale Lagrangian AKSZ with the source  $(T[1]X,\dx)$, the ghost degree of the target space symplectic structure and $\dim{X}$ are related as $\gh{\omega^F}=\dim{X}-1$.
Of course, one can still repeat the AKSZ construction, resulting in the shifted analog of the BV data where the symplectic structure on $Smaps(T[1]X,F)$ has degree $\gh{\omega^F}-\dim{X}$ while the analog of the BV-AKSZ action carries degree $\gh{\omega^F}-\dim{X}+1$. In particular, if $\gh{\omega^F}=\dim{X}$ the shifted AKSZ sigma model defines a Hamiltonian BFV system. Moreover, its associated BV formulation can be immediately constructed as an AKSZ sigma model with $X$ replaced with $X\times \fR^1$ (time-line), as was originally observed in~\cite{Grigoriev:1999qz,Barnich:2003wj}. The converse is also true, given an AKSZ sigma model on $T[1]X\times T[1]\fR^1$, its BFV Hamiltonian description is again a BFV-AKSZ sigma model with source $(T[1]X,\dx)$, see~\cite{Barnich:2003wj,Grigoriev:2010ic}. In a more general setup of not necessarily AKSZ systems analogous relations were thoroughly studied in~\cite{Cattaneo:2012qu,Cattaneo:2015vsa}, where the BFV formulation is interpreted as defined on the boundary of the spacetime.

If we now switch to presymplectic AKSZ systems, most of the above remains correct but there are some subtleties related to the regularity of the involved presymplectic structures. If the target space presymplectic structure is regular, the system is equivalent to the standard AKSZ model, at least locally. Indeed, inspecting the construction of \cite{Grigoriev:2022zlq,Dneprov:2024cvt} one finds that the resulting BV formulation is that of the symplectic AKSZ sigma-model whose target is the symplectic quotient of the initial target (we disregard global geometry issues in this discussion). The nontrivial examples of presymplectic BV-AKSZ models arise when the presymplectic structure is not regular but is such that the induced BV symplectic structure is regular provided we exclude the degenerate field configurations. For instance, in the presymplectic BV-AKSZ formulation of gravity-like models~\cite{Grigoriev:2020xec,Dneprov:2022jyn,Dneprov:2024cvt} among the AKSZ fields one typically has the frame field which should be assumed nondegenerate. We say that presymplectic structure is quasi-regular if the induced presymplectic structure on the space of AKSZ fields is regular for admissible configurations. It is natural to require that addmissible configurations are defined in such a way that a small variation of an admissible configuration is also admissible. Although we formulate these requirements in terms of functional space $Smaps(T[1]X,F)$ the role of the base $X$ is passive and the analysis boils down to (locally) finite dimensional space $Smaps(T_x[1] X,F)$ for a generic point $x\in X$.

Just like in the case of usual AKSZ models one can also consider (possibly shifted) presymplectic BV-AKSZ systems with $\gh{\omega^F}$ generic. For instance, for $\gh{\omega^F}=\dim{X}$ the presymplectic AKSZ system defines a local BFV system provided $\omega^F$ is quasi-regular. Because we have in mind applications to supersymmetric field theories, we generally do not have a natural integration on the source space and hence a natural induced  presymplectic structure. However, if we are only interested in the description at the level of equation of motion we can replace presymplectic structure by its kernel distribution $\cK$ which is by definiton a distribution generated by the vector fields annihilating the presymplectic structure. In contrast to the presymplectic  structure, $\cK$ naturally induces an associated distribution (also called prolongation of $\cK$) on the space of AKSZ fields. If the induced distribution is regular one can pass to the quotient space which is naturally a (shifted) local BV system at the level of equations of motion, see~\cite{Grigoriev:2024ncm} for a systematic exposition. In the case where the induced presymplectic structure is present, this is equivalent to the standard prescription, i.e. when the local BV system arises on the symplectic quotient of the space of (admissible) configurations.   

It is natural to address the question of equivalence between a presymplectic AKSZ system and its extension by replacing the presymplectic structures with their associated kernel distributions. Of course, in so doing we forget about Lagranians/symplectic structures and speak about equivalence at the level of equations of motion. More specifically, we define a weak AKSZ system $(F,Q,\cK,T[1]X)$ by taking $(T[1]X,\dx)$ as a source and a weak $Q$-manifold $(F,Q,\cK)$ as the target.  Recall, that a weak $Q$-manifold is an almost $Q$-manifold equipped with an involutive distribution $\cK$ such that:
\begin{equation}
L_Q\cK \subset \cK\,, \qquad Q^2\in \cK\,.
\end{equation}
In addition, the distribution $\bar\cK$ induced by $\cK$ on admissible configurations from $Smaps(T[1]X,F)$ is required to be regular. 
The physical interpretation of a weak AKSZ system is that of the BV system on the quotient of $Smaps(T[1]X,F)$ by $\bar \cK$. The BV-BRST differential on the quotient is induced by the lift of $\dx+Q$ to $Smaps(T[1]X,F)$. Of course weak AKSZ is just a particular case of so-called weak gauge PDEs introduced recently in~\cite{Grigoriev:2024ncm} to which we refer for further details and proofs.

Given a weak AKSZ model $(F,Q,\cK,T[1]X)$ let us consider its extension to $T[1]X \times T[1]\cT$, where now $T$ is either $\fR$ or $\Pi\fR$ (i.e. odd line). We also allow $X$ to be a supermanifold. Just like above, the extension can be considered as a weak AKSZ model on $T[1]X$ with the target being $Smaps(T[1]\cT,F)$. The target space $Smaps(T[1]\cT,F)$  is naturally an almost $Q$-manifold equipped with the compatible distribution $\cK^\prime$: indeed $Q^\prime$ is the lift of $Q+\dT$ to $Smaps(T[1]\cT,F)$ while $\cK^\prime$ is the lift of $\cK$. In this way we arrive at the new weak AKSZ model with an infinite-dimensional target $(Smaps(T[1]\cT,F),Q^\prime,\cK^\prime)$.  

Suppose that for the initial weak AKSZ model $(F,Q,\cK,T[1]X)$ we have a notion of admissible configurations and $\bar\cK$ is regular on the space of admissible configurations. A configuration of the extended weak AKSZ model can be represented as an element of $Smaps(T[1]X \times T[1]\cT,F)$. We call it admissible if its restriction to any point of the body of $T[1]\cT$ is admissible as an element of $Smaps(T[1]X,F)$. For instance if $\xi^a,\gh{\xi^a}=1$ are coordinates on $F$ encoding the frame field $e^a_\mu$ and admissible configurations are such that $\hat\sigma^*(\xi^a)=e^a_\mu(x)\theta^\mu+\ldots$ with $e^a_\mu(x)$ invertible then the configuration of the extended model is admissible if $e^a_\mu(x,t)$ defined through $\hat\sigma^*(\xi^a)=e^a_\mu(x,t)\theta^\mu+e^a_0(x,t)\theta^\prime+\ldots$ is invertible for all $t$ (in this example we assumed that $t$ is even).
We have the following:
\begin{prop}
Let $(F,Q,\cK,T[1]X,\dx)$ be a weak AKSZ model such that the induced distribution $\bar\cK$ is regular on admissible configurations. Let also $\cK^\prime$ be the distribution induced by $\cK$ on $Smaps(T[1]\cT,F)$ then the distribution $\bar\cK^\prime$ induced by $\cK^\prime$ on $Smaps(T[1]X,Smaps(T[1]\cT,F))$ is regular on admissible configurations. 
\end{prop}
In other words the extended weak AKSZ model has a natural notion of addmissible configurations and the induced distribution is quasiregular.  The proof is based on representing the space of admissible configurations for the extended model as  $Smaps(T[1]\cT,Smaps_\mathrm{ad}(T[1]X,F))$, where in the target we take only admissible configurations. As $\bar\cK$ is regular on $Smaps_{\mathrm{ad}}(T[1]X),F)$ by assumption, distribution $\bar\cK^\prime$ is regular as well.

The almost $Q$-structure $Q^\prime$ on $Smaps(T[1]\cT,F)$ is given by the natural lift of differential $d_T+Q$ to $Smaps(T[1]\cT,F)$. If, as before, we use coordinates $\psi^A(t,\theta^\prime)$ and assume $t$ fermionic,  we have $\psi^A(t,\theta^\prime)=\psi^A_0(\theta^\prime)+t \psi^A_1(\theta^\prime)$. In terms of $\psi^A_0(\theta^\prime)$ and $\psi^A_1(\theta^\prime)$ the action of $Q^\prime$ reads as:
\begin{equation}
Q^\prime \psi^A_0 (\theta^\prime)=\theta^\prime \psi^A_1(\theta^\prime)+Q^A(\psi_0(\theta^\prime))\,, \qquad 
Q^\prime \psi^A_1 (\theta^\prime)= \psi_1^B(\theta^\prime)\ddl{Q^A}{\psi^B}\Big|_{\psi=\psi_0(\theta^\prime)}.
\end{equation}
Because functions in $\theta^\prime$ are polynomials it is clear that all the coefficients in the expansion of $\psi_1$ in $\theta^\prime$ and all the coefficients in $\psi_0(\theta^\prime)$ save for $\psi_0(0)$ form contractible pairs. Moreover, because by assumption $\overbar{\cK}$ is regular one can assume that we work in the adapted coordinate system and hence the elimination can be performed after taking the quotient. After the elimination we are back to the local BV system determined by the initial weak AKSZ model. This shows that the extended weak AKSZ system (seen as a weak AKSZ model with the same source and the target $Smaps(T[1]\cT,F)$) is equivalent to the initial one as a local gauge theory. 

It is worth giving the following corollary of the above considerations. Let $(F,Q,\omega^F),(T[1]X,\dx)$ be a presymplectic BFV-AKSZ sigma model (i.e. $\gh{\omega^F}=\dim{X}$) and $\omega^F$ be quasi-regular. Consider the BV-AKSZ model obtained by extending the source $T[1]X$ to $T[1](X\times \fR^1)$. It defines a local BV system that is equivalent to the BV reformulation of the initial BFV system encoded in $(F,Q,\omega^F),(T[1]X,\dx)$. This statement can be considered as an extension of the analogous relation~\cite{Grigoriev:1999qz,Barnich:2003wj} known in the context of standard (symplectic) AKSZ models. The proof is based on the fact that for a regular presymplectic structure taking the symplectic quotient of $\bar F$ and passing to $Smaps(T[1]\fR^1,\bar F)$ commutes.     

\subsection{Generalization}\label{sec:generalAKSZ}

In fact adding/removing space-time dimension for the AKSZ model is a particular example of a slightly more general equivalence. First of all, we recall that a nonlagragian AKSZ model can be defined for a generic source $Q$-manifold $(\manX,\delta)$. If in addition $\manX$ is equipped 
with a $\delta$-invariant volume form of degree $-k$ and target space $F$ is equipped with a symplectic structure of degree $k-1$, the AKSZ construction still produces a local BV system. Of course, this has a straightforward generalisation to the presymplectic case.

More generally, one can replace $(\cC^\infty(\manX),\delta)$ with a rather general differential graded commutative algebra (DGCA) $(\algA,\delta)$. In reasonable cases the space of supermaps from $Spec(\algA)$ to $\cF$ (we keep using geometrical language and think of the space of super-homomorphsims from $\cC^{\infty}(F) \to \algA$ as the space of supermaps from $Spec(\algA)$ to $F$) is a graded supermanifold and $\delta+Q$ makes it into a $Q$-manifold, which in turn, defines a BV system. The Lagrangian version is constructed by taking $\algA$ equipped with a non-degenerate and $\delta$-invariant trace $\tr$ and $F$ with a symplectic structure of a compatible degree. This generalisation was proposed in~\cite{Bonechi:2009kx}, see also \cite{costello2011renormalization,Bonechi:2022aji,Ben-Shahar:2024dju}.

As we are first of all interested in local field theories we restrict ourselves to the case where $\algA$ is a module over $\cC^\infty(X)$ for some smooth (super)manifold $X$ which eventually plays the role of the (super)spacetime. In addition we also assume that seen as a $\cC^\infty(X)$-module, $\algA$ is a space of sections of a vector bundle over $X$ whose typical fiber is a graded commutative algebra $\mathfrak{a}$.  Note that $\delta$ is a 1st order differential operator on $X$. 

The discussed above relations between AKSZ sigma model on $T[1] X$ and its extension to  $T[1](X\times \cT)$, where $\cT$ is either $\fR^1$ or $\Pi\fR^1$, can be generalised as follows. Suppose that $\algA$ admits an additional degree such that the expansion of $\delta$ into homogeneous components reads as
\begin{equation}
\delta=\delta_{-1}+\delta_{0}+\delta_{1}+\ldots\,,
\end{equation}
and that $\delta_{-1}$ is $\cC^\infty(X)$-linear and constant rank. It follows,  cohomology of $\delta_{-1}$ is again a vector bundle and its space of sections is a graded commututative algebra $\tilde\algA$. Moreover, the differential $\delta$ induces a differential $\tilde \delta$, making $\tilde\algA$ into a differential graded commutative algebra (DGCA). One can derive $\tilde\delta$ using the usual spectral sequence arguments or as a version of equivalent reduction for the first-quantized BRST systems, see~\cite{Barnich:2004cr} for more details. It follows, the generalised AKSZ models with sources $(\algA,\delta)$ and $(\tilde\algA,\tilde\delta)$ are equivalent in the sense of elimination of generalized auxiliary fields.

Let us now turn to Lagrangian systems. Let $(\algA,\delta)$ be equipped with the trace $\tr :\algA \to \fR$ that carries ghost degree $-k$ and satisfies $\tr(\delta f)=0$. If the target space is equipped with the
$Q$-invariant presymplectic structure of degree $k-1$, the space of supermaps $Spec(\algA) \to F$ inherits the degree $-1$ presymplectic structure compatible with the total $Q$-structure which the lift of $\delta+Q$ to the space of supermaps. Assuming the induced presymplectic structure to be regular and taking a symplectic quotient of the space of supermaps gives a local BV system. In this way we have arrived at the generalisation of the presymplectic BV-AKSZ formulation, which is going to be useful in describing supersymmetric systems.  Note that even if the target space presymplectic structure is nondegenreate the presymplectic structure on the space of supermaps can degenerate because of the zero modes of the trace so that even in this case the system is to be interpreted as a presymplectic AKSZ.  

Let us present an explicit coordinate form of the main structures. For simplicity let us assume that $\algA$ is formed by functions on $X$ with values in the graded commutative algebra $\mathfrak{a}$ which we assume to be a subquotient of the polynomial algebra.  Moreover we assume that the trace factorises as $\tr{f}=\int d^nx \tr_0(f)$, where $\tr_0$ is the trace on $\mathfrak{a}$. If $e_i$ denotes a basis in $\mathfrak{a}$ a generic section $a \in \algA$ can be written as $f=f^i(x)e_i$. The trace defines a graded symmetric bilinear form $k_{ij}=\tr_0(e_ie_j)$. Now, if $\psi^A$ are coordinates on $F$, a generic supermap $Spec(\algA) \to F$ can be written as
\begin{equation}
\hat\sigma^\ast(\psi^A)=\psi^{Ai}(x)e_i\,, \qquad \gh{\psi^{Ai}}=\gh{\psi^A}-\gh{e_i}\,, \quad \p{\psi^{Ai}}=\p{\psi^A}-\p{e_i}\,\, \mathrm{mod} \,2 \,.
\end{equation}
It is convenient to identify the space of supermaps $Smaps(Spec(\algA),F)$ with $Smaps(X,\bar F)$ where $\bar F$ is the space of supermaps from $Spec(\mathfrak{a})$ to $F$. Note that $\psi^{Ai}$ can be identified with coordinates on $\bar F$. In this coordinate system the expression for the induced presympelctic structure at $\hat\sigma$ reads as:
\begin{equation}
\omega^{\bar F}_{\hat\sigma}=\tr_0\left(\hat\sigma(\omega^F_{AB}) d\psi^{Ai}e_i d\psi^{Bj}e_j\right)\,.
\end{equation}
In particular, if $\omega^F_{AB}=const$ one gets $\omega^{\bar F}=k_{ij}\omega^F_{AB}d\psi^{Ai} d\psi^{Bj} $. Note that if the basis $e_i$ splits into $(e_a,e_\alpha)$ such that $\tr_0(e_ae_b)=0=\tr_0(e_\alpha e_b)$ and
$k_{\alpha\beta}\equiv \tr_0(e_\alpha e_\beta)$ is invertible then the linear span of $e_a$ is an ideal in $\alga$. It is clear that all the fields $\psi^{Ab}$ associated to the ideal will be in the kernel of the presymplectic sructure. Of course there can be extra degeneracy originating from the kernel distribution of $\omega^F$.

The above considerations suggest that the general situation is well described by the following generalization of presymplectic BV-AKSZ sigma models: the base is a DGCA $(\algA,\delta)$  which is a module over $\cC^\infty(X)$ and the target is a weak presymplectic $Q$-manifold $(F,Q)$. It follows $\delta$ preserves the ideal in $\algA$ generated by the zero modes of $\tr$ and hence is well defined on the quotient by the ideal. The quotient is then a DGCA with the nondegenerate trace. In fact, it is enough to require that the image of $\delta^2$ belongs to the ideal. This is analogous to requiring $i_{Q^2}\omega^F=0$ in the presymplectic context. 

We see that there are two alternative approaches to describe a Lagrangian system in terms of the extended (i.e. associated to superspace) source space. The first one is to define the nondegenerate $\tr$ only on the equivalent reduction of $\algA$ to the cohomology of $\delta_{-1}$ while the second is to define the degenerate trace on $\algA$ in such a way that the quotient of $\algA$ by the ideal determined by the trace coincides with the cohomology of $\delta_{-1}$. In the second case, instead of taking the quotient by the ideal one can  immediately construct the presymplectic BV formulation on $Smaps(Spec(\algA),F)$ and then take the symplectic quotient to arrive at the underlying local BV system. We illustrate both approaches in the example of  superspace formulation of the supersymmetric particle in Section~\bref{sec:particle}.


\section{Presymplectic BV-AKSZ and rheonomy}
\label{sec:rheonomy}
\subsection{Palatini-Cartan-Weyl gravity}
The presymplectic BV-AKSZ formulation of gravity reviewed in Section~\bref{pAKSZ-Gravity} can be lifted to higher-dimensional space-time manifolds if one considers it at the level of equations of motion. More precisely, one reinterprets the presymplectic AKSZ model as a weak AKSZ model for which adding an extra contractible spacetime direction is merely an equivalence.  
Of course, if one is only interested in the equations of motion this is almost obvious because the equations of motion can be written in terms of the de Rham differential and the wedge product of spacetime differential forms and hence can be formally considered on any spacetime, see e.g.~\cite{Castellani:1991et}. We now uplift the presymplectic-AKSZ formulation described in section \bref{pAKSZ-Gravity} from $n$-dimensional spacetime to an extended space $P=X\times \fR^{n(n-1)/2} \to X$. The second factor will be then identified with the neighborhood  of the identity in the Lorentz group $SO(n-1,1)$. The adapted coordinates on $T[1]P$ are denoted as $(x^\mu, \atheta^\mu, x^i, \atheta^i)$ and the de Rham differential as $\dP$.

Let $\sigma: T[1]P \to \mathfrak{iso}(1,n-1)[1]$ be a given configuration satisfying the equations of motion. In local coordinate terms it reads as 
\begin{equation}
    \sigma^\ast(\xi^a) = \tilde{e}^a = \tilde{e}^a_\mu(x^\nu,x^i)\atheta^\mu + \tilde{e}^a_i(x^\nu,x^i)\atheta^i, \quad \sigma^\ast(\rho^{ab}) = \tilde{\omega}^{ab} =  \tilde{\omega}^{ab}_\mu (x^\nu,x^i)\atheta^\mu + \tilde{\omega}^{ab}_i(x^\nu,x^j)\atheta^i,
\end{equation}
and we assume  that $(\tilde{e}^a,\,\,\tilde{\omega}^{bc})$ define a nondegenerate coframe on $P$. Equations of motion \eqref{pAKSZ-eq} take the following form:
\begin{equation}\label{pAKSZ-eq-gravity-ext}
    \mathcal{V}_{abc}(e)\tilde{R}^{bc} =0, \quad \mathcal{V}_{abc}(e)\tilde{T}^c = 0,
\end{equation}
where $\tilde{R}^{ab} = \dP\tilde{\omega}^{ab} + \tilde{\omega}^a_{\enspace c}\tilde{\omega}^{cb}$ and  $\tilde{T^a} = \dP \tilde{e}^a + \tilde{\omega}^a_{\enspace b}\tilde{e}^b$. Expanding 2-forms $\tilde{R}^{ab}$ and $\tilde{T^a}$ in the coframe $(\tilde{e}^a, \tilde{\omega}^{ab})$ as
\begin{equation}
\begin{gathered}
    \tilde{R}^{ab} = \tilde{R}^{ab}_{\enspace\enspace cd|}\tilde{e}^c \tilde{e}^d + \tilde{R}^{ab}_{\enspace\enspace c|de} \tilde{e}^c \tilde{\omega}^{de} + \tilde{R}^{ab}_{\enspace\enspace |cdef}\tilde{\omega}^{cd}\tilde{\omega}^{ef},\\
    \tilde{T}^{a} = \tilde{T}^{a}_{\enspace cd|}\tilde{e}^c \tilde{e}^d + \tilde{T}^{a}_{\enspace c|de} \tilde{e}^c \tilde{\omega}^{de} + \tilde{T}^{a}_{\enspace |cdef}\tilde{\omega}^{cd}\tilde{\omega}^{ef}
\end{gathered}
\end{equation}
and substituting them into \eqref{pAKSZ-eq-gravity-ext} one finds:
\begin{subequations}
\begin{equation}\label{gravity-einst-eq}
    \tilde{R}^{ac}_{\enspace\enspace bc} - \frac{1}{2}\delta^a_b\tilde{R}^{cd}_{\enspace\enspace cd} = 0, \quad \tilde{T}^{a}_{\enspace cd|} = 0,
\end{equation}
\begin{equation}\label{gravity-rheonomy-eq}
    \tilde{R}^{ab}_{\enspace\enspace c|de} = \tilde{R}^{ab}_{\enspace\enspace |cdef} = \tilde{T}^{a}_{\enspace c|de} = \tilde{T}^{a}_{\enspace |cdef} =0\,,
\end{equation}
\end{subequations}
so that $\tilde{T}^a = 0$.

Let us consider a distribution $\mathcal{R}$ generated by all the vector fields annihilating $e^a$, i.e. vector fields $V$ satisfying $i_Ve^a=0$. In fact, $\mathcal{R}$ is involutive because $\sigma$ is a solution. Indeed, for any $X, Y \in \mathcal{R}$ one has:
\begin{equation}
\begin{gathered}
       e^a([X,Y]) = e^a(\mathcal{L}_XY) = \mathcal{L}_X\left(e^a(Y)\right) - (\mathcal{L}_Xe^a)(Y) = -i_Y(di_X + i_Xd)e^a =\\
       = -i_Yi_X (T^a-\omega^a_{\enspace b}e^b) = -i_Yi_X T^a = 0. 
\end{gathered}
\end{equation}
It follows $P$ is foliated by the integral submanifolds of $\mathcal{R}$ and hence locally $P$ is naturally a fiber bundle over $X$. Note that the bundle structure is determined by the choice of solution $\sigma$.

It is convenient to chose coordinates $x^\mu,x^i$ adapted to the bundle structure, i.e. $x^\mu$ are pullbacks of the coordinates on $X$ and $x^i$ are fiber coordinates. Note that in this coordinate system $e^a_i = 0$.
Now it is clear that equations  \eqref{gravity-einst-eq} pulled back to any section $X\to P$ is precisely the Einstein equation and the zero torsion condition. At the same time equations \eqref{gravity-rheonomy-eq} set to zero the remaining components of the curvature and the torsion and are known as the rheonomy conditions in the group manifold approach to (super)gravity \cite{DAuria:1982mkx,Castellani:1991et}. The pullback of $\omega^{ab}$ to a fiber of $P$ gives a nondegenerate  flat $so(n-1,1)$-connection and hence one is given with the $so(n-1,1)$ action on the fibers. In other words, locally we are dealing with the principle $SO(n-1,1)$-bundle over $X$. The $iso(n-1,1)$-valued 1-form defines a Cartan connection on $P$ and hence the structure of Cartan geometry modeled on $ISO(n-1,1)/SO(n-1,1)$. A recent discussion of the Cartan geometry perspective on supergravity can be found in \cite{Francois:2024rfh}. For a systematic exposition of Cartan geometry see e.g.~\cite{sharpe2000differential,Wise2009,cap2009parabolic}.



\subsection{N=1 D=4 Supergravity}

Let us now apply the analogous lift to the presymplectic BV-AKSZ formulation of supergravity from Section~\bref{sec:sugra}. Although one could again uplift the theory to the corresponding ``soft'' supergroup, it is more instructive and economical to uplift to the superspace which in the flat case can be identified with   a coset of the super-Poincar\'e group by the Lorentz subgroup. More precisely, consider superspace $M = \mathbb{R}^{4|4}$ with coordinates $(x^\mu, \stheta^\alpha)$. Let $\sigma: T[1]M \to \mathfrak{siso(1,3)}[1]$ be a fixed section: 
\begin{equation}
\begin{gathered}
    \sigma^\ast(\psi^\alpha) = \tilde{\psi}^\alpha = \tilde{\psi}^\alpha_\mu (x^\nu, \stheta^\beta) \atheta^\mu +  \tilde{\psi}^\alpha_{\enspace\gamma} (x^\nu, \stheta^\beta)\astheta^\gamma,\\
    \sigma^\ast(\xi^a) =\tilde{e}^a= \tilde{e}^a_\mu (x^\nu, \stheta^\beta) \atheta^\mu +  \tilde{e}^a_\gamma (x^\nu, \stheta^\beta)\astheta^\gamma,\\
    \sigma^\ast(\rho^{ab}) = \tilde{\omega}^{ab} = \tilde{\omega}^{ab}_\mu (x^\nu, \stheta^\beta) \atheta^\mu +  \tilde{\omega}^{ab}_\gamma (x^\nu, \stheta^\beta)\astheta^\gamma,
\end{gathered}
\end{equation}
such that $(\tilde{e}^a,\tilde{\psi}^\alpha)$ defines a nondegenerate coframe on $M$. Equations of motion \eqref{pAKSZ-eq} for $\sigma$ take the following form:
\begin{equation}\label{pAKSZ-sugra-eq-ext}
\begin{gathered}
    \tilde{R}^a =  \tilde{T}^a - \frac{1}{2} \overbar{\tilde{\psi}} \gamma^a \tilde{\psi} = 0,\\
   2\tilde{R}^{ab} \tilde{e}^c \epsilon_{abcd} - 4i\overbar{\tilde{\psi}} \gamma_5\gamma_d \tilde d_\Gamma \tilde{\psi} = 0,\\
    8\gamma_5 \gamma_a \tilde d_\Gamma \tilde{\psi} \tilde{e}^a + 4\gamma_5 \gamma_a \tilde{\psi} \tilde{R}^a = 0.
\end{gathered}
\end{equation}
Expanding the curvature 2-forms $R^a,R^{ab},R^\alpha\equiv (\tilde d_\Gamma\tilde \psi)^\alpha$ in the basis of  $(\tilde{e}^a,\tilde{\psi}^\alpha)$ the equations of motion \eqref{pAKSZ-sugra-eq-ext} give the spacetime equations of motion and the rheonomy conditions~\cite{Castellani:1991eu}.


\section{$\mathcal{N}=1$ supersymmetric particle and presymplectic BV-AKSZ}
\label{sec:particle}
In this section we use a toy-model to demonstrate how the presymplectic BV-AKSZ formalism can be employed to reconstruct the superspace Lagrangians and their BV extensions in a geometrical way. 

\subsection{$\mathcal{N}=1$ supersymetric mechanics}
Our toy-model is the  $\mathcal{N}=1$ superparticle. The respective SUSY algebra is:
\begin{equation}
    \{\mathcal{Q}, \mathcal{Q}\} = 2\mathcal{P}.
\end{equation}
It can be realised via the following vector fields 
\begin{equation}
    \mathcal{Q} = \partial_\stheta + \stheta \partial_t, \quad \mathcal{P} = \partial_t
\end{equation}
on the superspace $\mathbb{R}^{1|1}$ with coordinates $(t, \stheta)$. Moreover, on $\mathbb{R}^{1|1}$ one also introduces chiral derivative:
\begin{equation}\label{Chiral derivative}
    \D = \partial_\stheta - \stheta\partial_t, \quad \{\mathcal{Q}, \D\} = 0.
\end{equation}

The action for the particle model can be written in a manifestly supersymmetric way in terms of superfields $X^a(t,\stheta)$:
\begin{equation}\label{N=1 SUSY QM superspace action}
    S[X] = -\frac{1}{2}\int_{\mathbb{R}^{1|1}} dt d\stheta (\dot{X}_a \mathcal{D} X^a),
\end{equation}
where $\dot{X}^a \equiv  \partial_t X^a$. It is invariant under susy transformations 
\begin{equation}\label{N=1 SUSY QM superspace SUSY transform}
    \delta X^a(t,\stheta) = \varepsilon \mathcal{Q}X^a(t,\stheta) = \varepsilon( \partial_\stheta + \stheta \partial_t)X^a\,,
\end{equation}
where $\varepsilon$ is a constant fermionic parameter. In terms of the component fields introduced via
\begin{equation}
    X^a(t,\stheta) = x^a(t) + \stheta \lambda^a(t),
\end{equation}
the above transformations read as:
\begin{equation}\label{N=1 SUSY QM component SUSY transform}
    \delta x^a = \varepsilon \lambda^a, \quad \delta \lambda^a = -\varepsilon\dot{x}^a\,.
\end{equation}
Finally, performing integration over $\stheta$ one obtains the component action:
\begin{equation}\label{N=1 SUSY QM component action}
    S[x,\lambda] = \frac{1}{2}\int dt  \left(\dot{x}^2 + \lambda_a\dot{\lambda}^a\right)\,.
\end{equation}
Further details and examples on supersymmetric (quantum) mechanics can be found in e.g.~\cite{Smilga:2020nte}.

\subsection{AKSZ model for $\mathcal{N}=1$ supersymmetric particle}

The supersymmetry invariance of the model from the previous subsection can be gauged by coupling it to the 1d supergravity. The resulting model is usually called $\mathcal{N}=1$ supersymmetric particle and is most naturally constructed in the Hamiltonian BFV formalism by introducing ghosts associated to 1d supersymmetry transfomations. In the Lagrangian picture these ghosts give rise to the einbein and gravitino fields in 1 dimension. Quantization of the fermionic phase-space coordinates gives a spinor representation so that the model indeed describes a spinning particle. The field theoretical interpretation of the quantized model is that it gives a first-quantized description of the massless spinor field in Minkowski space~\cite{Henneaux:1987cpbis,Howe:1988ft}. 

The phase space of the Hamiltonian constrained system is the super-extension of $T^*\fR^d$ by fermionic coordinates $\lambda^a$, equipped with the symplectic structure $dp_a dx^a - \frac{1}{2}d\lambda_a d\lambda^a$. The 1st class constraints generating the supersymmetry transformations are
\begin{equation}
\half p^2\equiv  \half p_ap^a\,, \qquad (p\cdot \lambda)\,.
\end{equation}
Here and below we use $(A\cdot B)\equiv A^cB_c$ to denote the invariant pairing of Minkowski space (co)vectors. Introducing the respective ghost variables $\xi,\epsilon(\xi)=0$ and $\psi,\epsilon(\psi)=1$ as well as their canonically conjugated ghost momenta $\cP_\xi$ and $\cP_\psi$, the BRST charge of the system can be written as:
\begin{equation}
H = -\frac{p^2}{2}\xi - \psi (p\cdot\lambda) - \mathcal{P}_\xi \psi\psi\,.
\end{equation}
Together with the extended symplectic structure $\omega^F=d\chi$, $\chi=p_a dx^a - \frac{1}{2}\lambda_a d\lambda^a + \mathcal{P}_\xi d\xi - \mathcal{P}_\psi d\psi$ these data defines a degree zero symplectic $Q$-manifold $(F,Q,\omega^F)$, where $Q$ is determined by $i_Q\omega^F+dH=0$ and is given explicitly by:
\begin{equation}
\begin{gathered}
    Qx^a = p^a\xi + \psi \lambda^a,\qquad
    Qp^a = 0,\qquad 
    Q\lambda^a = -\psi p^a,\\
    Q\xi = \psi \psi,\qquad Q\psi =0,\\
    Q\mathcal{P}_\xi = \frac{1}{2}p^2,\qquad
    Q\mathcal{P}_\psi = (\lambda\cdot p) + 2\mathcal{P}_\xi\psi\,.
\end{gathered}
\end{equation}

The BV formulation of the system is given by the AKSZ sigma model with source $T[1]\fR^1$ and target $(F,Q,\omega^F)$.
Let $t,\theta$, $\gh{\theta}=1$ denote the source space coordinates so that $\dx=\theta\dl{t}$. Introducing notations for the ghost degree preserving supermaps  $\sigma: T[1]\mathbb{R}^{1|0} \to F$ via
\begin{equation}
\begin{gathered}
    \sigma^\ast(p^a) = p^a(t), \quad \sigma^\ast(x^a) = x^a(t), \quad \sigma^\ast(\lambda^a) = \lambda^a(t),\\
        \sigma^\ast(\xi) = \atheta e(t), \quad \sigma^\ast(\psi) = \atheta\psi(t), \\
        \sigma^\ast(\mathcal{P}_\xi) =\sigma^\ast(\mathcal{P}_\psi) = 0\,,
\end{gathered}
\end{equation}
the AKSZ action takes the following form:
\begin{equation}
\label{ext-hamiltonian}
S[x,p,\lambda,e,\psi]=\int dt\left((p\cdot\dot x)+\half (\lambda \cdot \dot \lambda) -e \frac{p^2}{2}-\psi (\lambda\cdot p) \right) \,.  
\end{equation}
This is of course the usual extended Hamiltonian action for the above constrained system. The einbein $e$ and gravitino $\psi$ enter the action as Lagrange multipliers associated to the respective constraints.

The initial model considered in the previous subsection is recovered by setting gauge fields $e,\psi$ to their background values: 
\begin{equation}\label{N=1-D=1-SUGRA-backgorund-worldline}
    e(t) = 1,\quad  \psi(t)=0\,.
\end{equation}
In particular, the gauge transformations preserving this background are the starting point supersymmetry transformations~\eqref{N=1 SUSY QM component SUSY transform} rewritten in the 1st order formulation. Indeed, these are given explicitly by
\begin{equation}
    \delta x^a(t) =  \varepsilon\lambda^a(t),  \quad \delta p^a(t) = 0, \quad \delta\lambda^a(t) = - \varepsilon p^a(t)\,.
\end{equation}
Eliminating $p$ via its own equations of motion one obtains SUSY transformations $\eqref{N=1 SUSY QM component SUSY transform}$. Finally, restricting AKSZ action \eqref{N=1-D=1-SUGRA-backgorund-worldline} to this background gives the standard 1st order reformulation of the superparticle action \eqref{N=1 SUSY QM component action}:
\begin{equation}\label{N=1-SQM-First-order-action}
    S[x,p,\lambda] = \int_{\mathbb{R}^{1|0}}dt\left((p\cdot\dot{x}) + \frac{1}{2}(\lambda \cdot \dot{\lambda}) - \frac{p^2}{2} \right).
\end{equation}

\subsection{Superfield formulation at the level of EOM from AKSZ}

We now employ the source space extension described in Section~\bref{sec:source-ext-eom}. Namely we replace $T[1]\fR^1$
with $T[1]X=T[1]\mathbb{R}^{1|1}=T[1]\fR^1 \times T[1](\Pi\fR^1)$. Denoting the coordinates on $\fR^{1|1}$ by $t,\eta$ and their associated fiber coordinates by $\atheta$ and $\astheta$ respectively the de Rham differential $\dx$ takes the form:
\begin{equation}
    \dx = \atheta\frac{\partial}{\partial t} + \astheta\frac{\partial}{\partial \stheta}\,.
\end{equation}
Let us also list the ghost degree and the fermionic degree of the coordinates:
\begin{equation}
\begin{gathered}
    \gh{t} = 0, \quad \epsilon(t) = 0, \qquad 
    \gh{\atheta} = 1, \quad \epsilon(\atheta) = 0,\\
    \gh{\stheta} = 0, \quad \epsilon(\stheta) = 1,\qquad 
    \gh{\astheta} = 1, \quad \epsilon(\astheta) = 1\,.
\end{gathered}
\end{equation}
For the future convenience we perform the following change of coordinates:
\begin{equation}\label{Superworldline new coordinates}
    t\to t, \quad \atheta\to \atheta+\stheta\astheta, \quad \stheta\to\stheta,\quad \astheta\to \astheta\,,
\end{equation}
so that in the new coordinates $\dx$ takes the following form:
\begin{equation}
\label{dxsusy}
    \dx=\atheta\partial_t +\astheta\D+\astheta^2 \partial_\atheta,
\end{equation}
where $\D$ is the chiral derivative defined in \eqref{Chiral derivative}. Note that the superspace analog of the
background solution \eqref{N=1-D=1-SUGRA-backgorund-worldline} in these coordinates read as:
\begin{equation}
\label{susy-background}
\sigma^*(\xi)=\theta, \qquad \sigma^*(\psi)=\astheta\,.
\end{equation}
Indeed, it solves the equations of motion in the sector of 
$\xi,\psi$-variables and has a
geometrical meaning of a Cartan connection 1-form describing $\mathbb{R}^{1|1}$ as a coset of the respective supergroup.

It is easy to see that $(\algA,\delta)$, where $\algA=\cC^\infty(T[1]\mathbb{R}^{1|1})$ and $\delta=\dx$, can be equivalently reduced to $(\tilde\algA,\tilde\delta)$ by employing the procedure explained in Section~\bref{sec:generalAKSZ}. Indeed, taking as a degree homogeneity in $\theta$ one finds that the lowest degree term in $\delta$ is $\delta_{-1}=\astheta^2 \partial_\atheta$ and hence $(\algA,\delta)$ equivalently reduces to the cohomology of $\delta_{-1}$. The cohomology can be explicitly identified as functions in $t,\eta,\astheta$ modulo the relation $\astheta^2=0$ and the induced differential is given by   $\tilde\delta=\astheta \D$. 


In terms of the reduced system the background solution \eqref{susy-background} takes the form:
\begin{equation}\label{N=1-D=1-SUGRA-backgorund-reduced-superworldline}
    \sigma^\ast(\xi) = 0, \quad \sigma^\ast(\psi) = \astheta\,.
\end{equation}
As before, we also set to zero ghost momenta $\mathcal{P}_\xi, \mathcal{P}_\psi$. The remainning freedom in the ghost degree preserving supermaps $\sigma^\ast: C^\infty(F)\to \tilde\algA$ is parameterized by:
\begin{equation}\label{N=1 D=1 SUSY QM supermaps reduced}
\begin{gathered}
    \sigma^\ast(x^a) = X^a(t,\stheta), \quad  \sigma^\ast(p^a) = P^a(t,\stheta), \quad \sigma^\ast(\lambda^a) = \Lambda^a(t,\stheta), \,.
\end{gathered}
\end{equation}
 The equations of motion are the $Q$-map condition $\dx\circ \sigma^*=\sigma^*\circ Q$ applied to the remaining coordinates $x,p,\lambda$ and are given by:
\begin{equation}
\begin{gathered}
    \astheta\D X^a(t,\stheta) = \astheta\Lambda^a(t,\stheta),\\
    \astheta \D \Lambda^a(t,\stheta) = -\astheta P^a(t,\stheta),\\
    \astheta\D P^a(t,\stheta) =0\,.
\end{gathered}
\end{equation}
They can be equivalently written as  
\begin{equation}\label{N=1 SUSY QM AKSZ Superspace EOM}
    \D \dot{X}^a(t,\stheta) = 0, \quad  \Lambda^a(t,\stheta) = \D X^a(t,\stheta), \quad P^a = -\D \Lambda^a(t,\stheta)\,.
\end{equation}
The first equation is nothing but the equation of motion of the superfield action~\eqref{N=1 SUSY QM superspace action} while the second and third express $P,\Lambda$ in terms of $X$ so that $P,\Lambda$ are the auxiliary superfields.

Let us finally check that the residual gauge transformation indeed reproduce susy transformations. Recall that gauge parameters are the ghost degree $-1$ vertical vector fields on $T[1]X\times F$ and hence they can be written as $Y = a(t,\stheta)\frac{\partial}{\partial \xi}+\alpha(t,\stheta)\frac{\partial}{\partial\psi}$. Demanding the corresponding gauge transformation $\delta\sigma=\commut{\dx+Q}{Y}$ to preserve background \eqref{N=1-D=1-SUGRA-backgorund-reduced-superworldline} one gets
\begin{equation}
    a(t,\stheta) = b -2 \stheta\beta \quad \alpha(t,\stheta) = \beta,
\end{equation}
where $b$ (resp. $\beta$) is the constant bosonic (resp. fermionic) parameter. Focusing on the fermionic symmetry and hence setting $b=0$ one finds:
\begin{equation}
    \delta X^a(t,\stheta) = -2\stheta\beta P^a(t,\stheta) +\beta \Lambda^a(t,\stheta), \quad \delta P^a(t,\stheta)=0, \quad \delta \Lambda^a(t,\stheta) = - \beta P^a(t,\stheta).
\end{equation}
Finaly, using \eqref{N=1 SUSY QM AKSZ Superspace EOM} one gets
\begin{equation}
    \delta X^a(t,\stheta) = \beta (\partial_\stheta + \stheta\partial_t) X^a(t,\stheta)\,,
\end{equation}
which is exactly \eqref{N=1 SUSY QM superspace SUSY transform}.

\subsection{Lagrangian superfield formulation via AKSZ}

As we discussed in Section~\bref{sec:generalAKSZ} there can be two equivalent ways to define Lagrangian AKSZ formulation in the case where the source space doesn't have a natural volume form. It turns out that both work perfectly for the toy-model in question.

The first approach is to equivalently reduce the source space to a DGCA $(\tilde\algA,\tilde\delta)$ that admits an invariant trace. In the case at hand this is precisely the reduced algebra $\tilde\algA=H^\bullet(\delta_{-1})$ introduced in the previous subsection. The induced differential therein is $\tilde\delta=\astheta \D$. Identifying $\tilde\algA$ as a tensor product of $\cC^\infty(X)$, $X=\fR^{1|1}$ and the two-dimensional algebra $\alga=\mathbb{R}[\astheta]/(\astheta^2)$ there is a natural trace operation of ghost degree $-1$ given by the tensor product of $\int dt d\eta$ on $\cC^\infty(\fR^{1|1})$ and $\tr_0$ on $\alga$ defined by
\begin{equation}\label{Reduced superworldline trace}
    \tr_0(1) = 0, \quad \tr_0(\astheta) = 1.
\end{equation}
Equivalently:
\begin{equation}
\tr\left( f_0(t,\eta)+\astheta f_1(t,\eta)\right)=\int dt d\eta f_1(t,\eta) \,.
\end{equation}
It easy to see that it is compatible with $\tilde\delta$ in the sense that 
\begin{equation}
    \int dt d \stheta Tr_0(\tilde\delta f(t,\stheta,\astheta)) = 0
\end{equation}
modulo boundary terms. Let us note that the analogous approach was proposed in \cite{Alkalaev:2018bqe}, see also~\cite{Chekmenev2021},
to define a correct inner product in the BFV-BRST quantization of systems with fermionic constraints. In fact in the simplest case of $N=1$ spinning particle model the inner product is precisely the above trace tensored 
with the usual inner-product for spinors.

Employing notations~\eqref{N=1 D=1 SUSY QM supermaps reduced} for the ghost degree preserving supermaps and setting $\sigma^*(\xi),\sigma^*(\psi)$ to their background values \eqref{N=1-D=1-SUGRA-backgorund-reduced-superworldline} the AKSZ-like action takes the form:
\begin{equation}\label{N=1-SQM-First-order-superfield-action}
    S[X,P,\Lambda] =\int dtd\eta \left((P\cdot\D X) - \frac{1}{2}(\Lambda \cdot \D \Lambda) - (P\cdot\Lambda)\right).
\end{equation}
Eliminating $P$ and $\Lambda$ from \eqref{N=1-SQM-First-order-superfield-action} using their own equations of motion
\begin{equation}
    \frac{\delta S}{\delta P^a(t,\stheta)}= \D X^a - \Lambda^a =0, \quad  \frac{\delta S}{\delta \Lambda^a(t,\stheta)} = -\D \Lambda^a -P^a = 0.
\end{equation}
one gets
\begin{equation}
    S = \frac{1}{2}\int_{\mathbb{R}^{1|1}} dtd\stheta(\dot{X}\cdot\D X), \quad \delta X^a(t,\eta) = \beta (\partial_\stheta + \stheta\partial_t) X^a(t,\stheta)
\end{equation}
which coincides with \eqref{N=1 SUSY QM superspace action} up to an overall sign. 

To obtain the full scale BV-formulation one needs to pass to supermaps $\sigmahat^\ast : C^\infty(F)\to \tilde{\mathcal{A}}$, i.e.
\begin{equation}
\begin{gathered}
    \sigmahat^\ast(x^a) = X^a(t,\stheta)+\astheta\overset{1}{X^a}(t,\stheta),\qquad \sigmahat^\ast(p^a) = P^a(t,\stheta)+\astheta\overset{1}{P^a}(t,\stheta), \\ \sigmahat^\ast(\lambda^a) = \Lambda^a(t,\stheta)+\astheta\overset{1}{\Lambda^a}(t,\stheta).
\end{gathered}
\end{equation}
The induced BV symplectic form read as 
\begin{equation}
    \omega_{BV} = \int dt d\eta \left(d\overset{1}{P_a}(t,\stheta)d X^a(t,\stheta) + d P^a(t,\stheta) d \overset{1}{X_a} - d\overset{1}{\Lambda_a}(t,\stheta)d\Lambda^a(t,\stheta)\right)\,.
\end{equation}
It follows the BV symplectic structure is nondegenerate and coincides with that of the usual BV extension of the superfield action~\eqref{N=1-SQM-First-order-superfield-action}.

Let us now turn to the second approach.  We saw in Section~\bref{sec:generalAKSZ} that instead of reducing $(\algA,\delta)$ in the first place, one can alternatively introduce a degenerate trace on $\algA$ in such a way that the induced presymplectic structure on the space of AKSZ fields is degenerate and its symplectic reduction gives the correct BV description. In our toy-model example this is easily done by defining the ghost degree $-1$ trace on $\algA$ as follows
\begin{equation}
    \tr_0(f(\atheta, \astheta)) = \int d\theta d\astheta f(\atheta, \astheta) \mathbb{Y}^{0|1}\,, \qquad \mathbb{Y}^{0|1} = \atheta \delta'(\zeta)\,,
\end{equation}
where  $\delta^\prime(\zeta)$ denotes the derivative of the Dirac delta-function.\footnote{The derivative of the Dirac delta function can be defined in purely algebraic terms by requiring $\int d \astheta \delta^\prime(\astheta)f(\astheta)=-f^\prime(0)$.} Note that 
$\mathbb{Y}^{0|1} = \atheta \delta'(\zeta)$ is precisely the picture changing operator from \cite{Castellani:2017ycm}. It is easy to see, that $\tr=\int dt d\eta \tensor \tr_0$ is compatible with $\dx$ and that all basis elements proportional to $\atheta$ or $\astheta^2$ are in the kernel.  Moreover, defining the presymplectic structure on the space of supermaps from $T[1]\fR^{1|1}$ to $F$ using the above trace one finds that all the fields associated to the basis elements proportional to $\atheta$ or $\astheta^2$ are in the kernel of the presymplectic structure and hence taking the symplectic quotient gives the same reduced space of supermaps as the reduction of the source to $(\tilde\algA,\tilde\delta)$. This gives an example of the presymplectic BV-AKSZ system where the presymplectic structure in the target space is nondegenerate while the induced presymplectic structure on the space of supermaps is degenerate thanks to the degenerate trace.


\subsection*{Acknowledgments}

We are grateful to  Pietro Grassi for collaboration on the initial stage of this project. Discussions with Ivan Dneprov are gratefully acknowledged.  MG also wishes to thank Nicolas Boulanger, Martin Cederwall, Alexey Kotov, Vasileios Letsios, Thomas Strobl and  Maxim Zabzine for useful exchanges. The work of AM was supported by the BASIS foundation scholarship.

\begin{appendix}
\section{Conventions and useful formulas}
\label{sec:app-conventions}
Symmetrization and antisymmetrization are defined with the factor of the inverse factorial. Minkowski metric and Levi-Civita tensor are 
\begin{equation}
    \eta_{ab} = diag(1,-1,-1,-1), \quad \epsilon_{0123} =1, \quad \epsilon^{0123}=-1.
\end{equation}
Grassmann parity of differential forms on a graded supermanifolds 
and their graded commutativity relations are defined as
\begin{equation}
    \p{a} = \gh{a} + \epsilon(a) + \mathrm{fdeg}(a) \enspace \mathrm{mod}\,\, 2, \quad ab = (-1)^{\p{a}\p{b}}ba,
\end{equation}
where $gh$ denotes the ghost degree, $\epsilon$ is the fermionic degree and $\mathrm{fdeg}$ is the form degree.  Interior and exterior derivatives are defined as
\begin{equation}
    d = dx^A\partial_A, \quad i_Q = Q^A\frac{\partial}{\partial dx^A}.
\end{equation}

Gamma matrices and charge conjugation matrices are:
\begin{equation}
\begin{gathered}
    \{\gamma^a,\gamma^b\} = 2\eta^{ab}, \quad 
    (\gamma^a)^\alpha_{\enspace\beta}=\begin{pmatrix}
        0 & \sigma^a\\
        \sigmabar^a & 0
    \end{pmatrix}, \\
    \gamma_{ab} = \gamma_{[a}\gamma_{b]}, \quad \gamma_5 = i\gamma^0\gamma^1 \gamma^2\gamma^3, \quad C_{\alpha\beta} = \begin{pmatrix}
        i\sigma^2 & 0\\
        0 & -i\sigma^2
    \end{pmatrix}.
\end{gathered}
\end{equation}

Indices of product of gamma matrices $\Gamma$ and spinors $\psi$ are lowered and raised as follows:
\begin{equation}
    \Gamma_{\alpha\beta} = C_{\alpha\gamma}\Gamma^\gamma_{\enspace\beta}, \quad \Gamma^{\alpha\beta} = C^{\beta\gamma}\Gamma^\alpha_{\enspace\gamma}, \quad \psi_\alpha = \psi^\beta C_{\beta\alpha}.
\end{equation}
For the Dirac conjugate of Majorana spinor $\psibar = \psi^\dagger\gamma^0 = \psi^T C$ we omit bar, if its index is written explicitly, i.e. $\psibar_\alpha\equiv \psi_\alpha$.

Properties of gamma matrices and their products:
\begin{equation}\label{Gamma matrix symmetry properties}
\begin{gathered}
    (\gamma^a)_{\alpha\beta} =  (\gamma^a)_{(\alpha\beta)}, \quad  (\gamma^5)_{\alpha\beta}=(\gamma^5)_{[\alpha\beta]}, \quad (\gamma^a\gamma^5)_{\alpha\beta}=(\gamma^a\gamma^5)_{[\alpha\beta]}, \quad (\gamma^{ab})_{\alpha\beta} =  (\gamma^{ab})_{(\alpha\beta)},\\
    (\gamma_a\gamma_{bc})_{\alpha\beta} = 2(\eta_{a[b}\gamma_{c]})_{(\alpha\beta)} + i \epsilon_{abcd}(\gamma_5\gamma^d)_{[\alpha\beta]}.
\end{gathered}
\end{equation}
Fierz identity reads as:
\begin{equation}\label{Fierz identity}
    \gamma_a\psi \psibar \gamma^a\psi = 0, 
\end{equation}
where $\psi$ is Majorana and $\p{\psi} = 0$.

    \section{Calculation of $i_Q \omega^F$ for supergravity}
    \label{app:iq=dH}
    Here we give detailed derivations of $i_Q \omega^F$: 
    \begin{equation}
    \begin{gathered}
        i_Q \omega^F = -\epsilon_{abcd}\xi^d \rho^a_{\enspace e} \xi^e d\rho^{bc}- \epsilon_{abcd}\xi^d d\xi^a \rho^b_{\enspace e} \rho^{ec} + \\
        + \frac{1}{2} \epsilon_{abcd} (\gamma^a)_{\alpha\beta} \xi^d \psi^\alpha \psi^\beta d \rho^{bc}  +i (\gamma_5\gamma_a)_{\alpha\beta}\xi^a (\gamma_{bc})^\alpha_{\enspace \gamma} \rho^{bc}\psi^\gamma d\psi^\beta  - \\
        -  \frac{i}{2}(\gamma_5\gamma_a)_{\alpha\beta} \psi^\alpha (\gamma_{cd})^\beta_{\enspace\gamma} \rho^{cd} \psi^\gamma d\xi^a - 2i (\gamma_5\gamma_a)_{\alpha\beta} \psi^\alpha d\psi^\beta (\rho_{\enspace b}^a \xi^b - \frac{1}{2} (\gamma^a)_{\gamma\delta} \psi^\gamma \psi^\delta).
    \end{gathered}
    \end{equation}
    \begin{itemize}
        \item First and second terms are the same as in Palatini-Cartan-Weyl gravity:
        \begin{equation}
           -\epsilon_{abcd}\xi^d \rho^a_{\enspace e} \xi^e d\rho^{bc}- \epsilon_{abcd}\xi^d d\xi^a \rho^b_{\enspace e} \rho^{ec} = -d\left(\frac{1}{2}\epsilon_{abcd}\xi^c \xi^d \rho^a_{\enspace e} \rho^{eb}\right)
        \end{equation}
        \item Using the last formula of \eqref{Gamma matrix symmetry properties}, the third and the fourth terms can be rewritten as:
        \begin{equation}
            \frac{1}{2} \epsilon_{abcd}(\gamma^a)_{\alpha\beta}\xi^d \psi^\alpha \psi^\beta d \rho^{bc} = \frac{i}{2}\psi^\alpha(\gamma_5 \gamma_a \gamma_{bc})_{\alpha\beta}  d\rho^{bc} \psi^\beta \xi^a,
        \end{equation}
        \begin{equation} \label{4 summand}
           i (\gamma_5\gamma_a)_{\alpha\beta}\xi^a (\gamma_{bc})^\alpha_{\enspace \gamma} \rho^{bc}\psi^\gamma d\psi^\beta =  id \psi^\alpha(\gamma_5 \gamma_a \gamma_{bc})_{(\alpha\beta)} \rho^{bc} \psi^\beta \xi^a  -2i (\gamma_5\gamma_a)_{\alpha\beta} \rho^a_{\enspace b}\xi^b d\psi^\alpha \psi^\beta.
        \end{equation}
        \item After relabeling indices and changing order of variables the fifth term takes the following form:
        \begin{equation}
            -\frac{i}{2}(\gamma_5\gamma_a)_{\alpha\beta} \psi^\alpha (\gamma_{cd})^\beta_{\enspace\gamma} \rho^{cd} \psi^\gamma d\xi^a = -\frac{i}{2}\psi^\alpha(\gamma_5\gamma_a\gamma_{bc})_{\alpha\beta} \rho^{bc}\psi^\beta d\xi^a
        \end{equation}
        \item  The sixth term can be rewritten as:
        \begin{equation}
            -2i(\gamma_5 \gamma_a)_{\alpha \beta} \psi^\alpha d\psi^\beta \rho^a_{\enspace b} \xi^b = 2i(\gamma_5\gamma_a)_{\alpha\beta} \rho^a_{\enspace b}\xi^b d\psi^\alpha \psi^\beta\,,
        \end{equation}
        so that it cancels one of the terms in \eqref{4 summand}.
        \item The last term vanishes thanks to Fierz identity \eqref{Fierz identity}.
    \end{itemize}
    
    Collecting summands and taking into account $\psi^\alpha \psi^\beta = \psi^{(\alpha} \psi^{\beta)}$ we get the result:
    \begin{equation}
        \begin{gathered}
        i_Q \omega^F = -d\left(\frac{1}{2}\epsilon_{abcd}\xi^c \xi^d \rho^a_{\enspace e} \rho^{eb}\right) + i(\gamma_5 \gamma_a \gamma_{bc})_{(\alpha\beta)} \left(\frac{1}{2}  \psi^\alpha d\rho^{bc}  \psi^\beta \xi^a +d\psi^\alpha  \rho^{bc}  \psi^\beta \xi^a- \frac{1}{2}\psi^\alpha \rho^{bc}  \psi^\beta d\xi^a\right) = \\
        =- d \left(\frac{1}{2}\epsilon_{abcd}\xi^c \xi^d \rho^a_{\enspace e} \rho^{eb} - \frac{i}{2} \psi^\alpha (\gamma_5 \gamma_a \gamma_{bc})_{\alpha\beta} \rho^{bc} \psi^\beta \xi^a\right).
        \end{gathered}
    \end{equation}

\section{Kernel of presymplectic form of N=1 D=4 Supergravity}\label{app:kernel of presympelctic form}
Here we introduce two-component spinor notations for $F$ and $\bar F$ and give a proof of Lemma~\bref{lemma:kernel}. For two-component spinor parametrization of spin-tensors in 4 dimensions we use the conventions from~\cite{huggett1994introduction}.

It is convenient to introduce the following basis in the Grassmann algebra generated by  $\xi$:
\begin{equation}
\begin{gathered}
    \xi_{A\alphadot} = (\sigma^a)_{A\alphadot}\xi_a, \quad \xi_{(A\B)} = \frac{i}{2}(\sigma_{ab})_{A\B}\xi^a\xi^b, \quad \xibar_{(\alphadot\betadot)} = -\frac{i}{2}(\sigmabar_{ab})_{\alphadot\betadot}\xi^a\xi^b,\\
    \xitilde_{A\alphadot} = (\sigma^{a})_{A\alphadot} \mathcal{V}_a(\xi), \quad \xitilde = \mathcal{V}(\xi).
\end{gathered}
\end{equation}
The nonvanishing products in this algebra are :
\begin{equation}
\begin{gathered}
    \xi_{A\alphadot}\xi_{\B \betadot} =  \xi_{A\B} \epsilon_{\alphadot\betadot}+\epsilon_{A\B}\xibar_{\alphadot\betadot}, \qquad  \xi_{A\alphadot}\xibar_{\betadot\gammadot} =  2i\xitilde_{A (\betadot}\epsilon_{\gammadot)\alphadot},\\
     \xi_{A\alphadot} \xi_{\B\C} = -2i\xitilde_{\alphadot(\B}\epsilon_{\C)A}, \qquad \xi_{A\alphadot}\xitilde_{\B\betadot}=2\epsilon_{A\B}\epsilon_{\alphadot\betadot}\xitilde,\\
    \xi_{A\B}\xi_{\C D} = 2i\xitilde(\epsilon_{A\C}\epsilon_{\B D}+\epsilon_{A D}\epsilon_{\B\C}), \qquad \xibar_{\alphadot\betadot}\xibar_{\gammadot\deltadot} = -2i\xitilde(\epsilon_{\alphadot\gammadot}\epsilon_{\betadot\deltadot}+\epsilon_{\alphadot\deltadot}\epsilon_{\betadot\gammadot}).\\
\end{gathered}
\end{equation}
Coordinates on $F=\algg[1]$ associated with Lorentz and susy transforsormations are parameterised as:
\begin{equation}
    \psi^\alpha=\begin{pmatrix}
    \psi_A\\
    \psibar^\alphadot\
\end{pmatrix},\quad \rho^{ab} = \frac{i}{4}\left(\rho_{A\B}(\sigma^{ab})^{A\B}-\rhobar_{\alphadot\betadot}(\sigmabar^{ab})^{\alphadot\betadot}\right).
\end{equation}

In terms of the new parametrization the presympelctic form can be written as:
\begin{equation}
\begin{gathered}
    \omega^F = \omega^F_0 + \omega^F_1\,,\\
    \omega^F_0=\frac{i}{2}\left(\xi_A^\gammadot d\xi_{\B\gammadot}d\rho^{A\B} - \xi_\alphadot^\C d\xi_{\betadot\C}d\rhobar^{\alphadot\betadot}\right)- 4id\psibar^\alphadot d\psi^A \xi_{A\alphadot} \,,\\
    \omega^F_1=-2i (d\psibar^\alphadot \psi^{A} - \psibar^\alphadot d\psi^A) d\xi_{A\alphadot},
\end{gathered}
\end{equation}
where the first term in $\omega^F_0$ is the presymplectic form for Palatini-Cartan-Weyl gravity while the remaining terms in $\omega^F$ is the presymplectic form in the gravitino sector decomposed into the homogeneous in $\psi$ components.

The spinor coordinates on $\bar F=Smaps(T[1]\mathbb{R}^{4|0},F)$ are introduced as follows:
\begin{equation}
    \sigmahat^\ast(\Psi^I) = \Psi^I +\Psi^I{}_{{A\alphadot}}\theta^{A\alphadot} + \overset{2}{\Psi}{}^I{}_{{A\B}}\theta^{A\B}+\overset{2}{\Psi}{}^I{}_{{\alphadot\betadot}}\thetabar^{\alphadot\betadot} + \overset{3}{\Psi}{}^I{}_{{A\alphadot}}\thetatilde^{A\alphadot} + \overset{4}{\Psi}{}^I\thetatilde,
\end{equation}
where $\Psi^I$ is a collective notation of coordinates on $F$. The induced presymplectic form $\bar\omega$ evaluated at $p\in\textrm{body}(\bar F)$ reads as
\begin{multline}
    \bar\omega_p = \bar\omega^{PCW}_p- 8i(d\psi^A d\overset{3}{\psibar^{\alphadot}}_{|A\alphadot}-d\psibar^\alphadot d\overset{3}{\psi^A}_{|A\alphadot}) -
    16 (d\psi_A^{\enspace|A\betadot}d\overset{2}{\psibar^\alphadot}_{|\alphadot\betadot}+d\psibar_\alphadot^{\enspace|\alphadot\B}d\overset{2}{\psi^A}_{|A\B}) 
    + \\
    +16( d\psi^{(A|\B)\alphadot}d\overset{2}{\psibar}_{\alphadot|{A\B}}+d\psibar^{(\alphadot|\betadot)A}d\overset{2}{\psi}_{A|\alphadot\betadot}).
\end{multline}
To give a proof of the Lemma~\bref{lemma:kernel} we need the explicit form of the vector fields on $F$ belonging to the kernel of $\omega^F_0$ and whose prolongation to $\bar F$ restricted to $p$ generate $\bar\cK_p \in T_p\bar F$. In the sector of $\xi^a$ variables these can be taken as:
\begin{equation}
\label{xi-vf}
    X_{A\alphadot} = \xitilde \frac{\partial}{\partial \xi^{A\alphadot}}, \qquad 
    X_3= 
    \xitilde^{\C\gammadot}\frac{\partial}{\partial \xi^{\C\gammadot}}, \qquad  X_{(A\B)(\alphadot\betadot)} = \xitilde_{A \alphadot}\frac{\partial}{\partial \xi^{\betadot\B}}+\dots,
\end{equation}
where we have singled out the irreducible components and $\dots$ denote the appropriate symmetrisation. We omit the analysis of the vector fields in $\rho$ sector as they do not require corrections.

For $\psi$-sector one can check that the kernel of $\bar\omega$ at point $p$ is given by the prolongation of the following vector fields:
\begin{equation}
\label{psi-vf}
\begin{gathered}
    Y_{(A\B\C)} = \xi_{A\B}\frac{\partial}{\partial \psi^\C} + \dots , \qquad \overbar{Y}_{(\alphadot\betadot\gammadot)}=\xibar_{\alphadot\betadot}\frac{\partial}{\partial \psi^\gammadot}  + \dots,\\
    Y_{(A\B)\alphadot} = \xitilde_{A\alphadot}\frac{\partial}{\partial \psi^\B}+ \dots, \qquad \overbar{Y}_{(\alphadot\betadot)A} = \xitilde_{A\alphadot}\frac{\partial}{\partial \psibar^\betadot}+ \dots,\\
    Y_A = \xitilde\frac{\partial}{\partial\psi^A}, \qquad   \overbar{Y}_\alphadot = \xitilde\frac{\partial}{\partial\psibar^\alphadot},
\end{gathered}
\end{equation}
These vector fields belong to the kernel of $\omega^F_0$.

To prove Lemma~\bref{lemma:kernel} it is enough to show that vector fields~\eqref{xi-vf}
and \eqref{psi-vf} can be modified by terms proportional to $\psi$ in such a way that the resulting vector fields belong to the kernel of the $\omega^F$. Indeed, vector fields of the form $V^{ab}\dl{\rho^{ab}}$ and belonging to the kernel of $\omega^F_0$ are automatically in the kernel of $\omega^F$ and hence there is no need for modifications.

In the $\xi$-sector the modified  vector fields are given by:
\begin{equation}
\begin{gathered}
    \tilde{X}_{A\alphadot} = \xitilde \frac{\partial}{\partial \xi^{A\alphadot}} + \frac{1}{4}\xitilde_{A\alphadot}\left(\psi^\B\frac{\partial}{\partial \psi^\B}+\psibar^\betadot\frac{\partial}{\partial \psibar^\betadot}\right),
    \\
    \tilde{X}_3 = \xitilde^{\C\gammadot}\frac{\partial}{\partial \xi^{\C\gammadot}} -\frac{i}{6}\psibar_\gammadot\xibar^{\gammadot\deltadot}\frac{\partial}{\partial \psibar^\deltadot} +  \frac{i}{6}\psi_\C\xi^{\C D}\frac{\partial}{\partial \psi^D} + \\
    +\frac{2i}{3}\psi_\C\psibar^\gammadot\xi_{D\gammadot}\frac{\partial}{\partial \rho_{\C D}}- \frac{2i}{3}\psibar_\gammadot\psi^\C\xi_{\deltadot\C}\frac{\partial}{\partial \rhobar_{\gammadot\deltadot}},\\
    \tilde{X}_{(A\B) (\alphadot\betadot)} = \xitilde_{A \alphadot}\frac{\partial}{\partial \xi^{\betadot\B}} - \frac{i}{4}\xi_{A\B}\psibar_\alphadot\frac{\partial}{\partial \psibar^\betadot} + \frac{i}{4}\xibar_{\alphadot\betadot}\psi_A\frac{\partial}{\partial \psi^\B}+
    \\
    +i\psibar_\alphadot\psi^\C\xi_{\betadot\C}\frac{\partial}{\partial \rho^{A\B}} - i\psi_A\psibar^\gammadot\xi_{\B\gammadot}\frac{\partial}{\partial \rhobar^{\alphadot\betadot}}+\dots,
\end{gathered}
\end{equation}
where $\dots$ denote appropriate symmetrisation. In the $\psi$-sector the necessary modification of \eqref{psi-vf} reads as:
\begin{equation}\label{Supergravity kernel psi-sector corrected}
\begin{gathered}
 \tilde{Y}_{(A\B\C)} = \xi_{A\B}\frac{\partial}{\partial \psi^\C} + 4\psibar^\betadot\xi_{\betadot\C}\frac{\partial}{\partial \rho^{A\B}} +\dots , \\
 \tilde{\overbar{Y}}_{(\alphadot\betadot\gammadot)}=\xibar_{\alphadot\betadot}\frac{\partial}{\partial \psibar^\gammadot}  +4\psi^\B\xi_{\B\gammadot}\frac{\partial}{\partial \rhobar^{\alphadot\betadot}}+ \dots,\\
 \tilde{Y}_{(A\B)\alphadot} = \xitilde_{A\alphadot}\frac{\partial}{\partial \psi^\B}-2i\psibar^\gammadot\xibar_{\gammadot\alphadot}\frac{\partial}{\partial \rho^{A\B}}-i \psibar_\alphadot \xi^\gamma_A \frac{\partial}{\partial \rho^{\C\B}}+ \dots, \\
 \tilde{\overbar{Y}}_{(\alphadot\betadot)A} = \xitilde_{A\alphadot}\frac{\partial}{\partial \psibar^\betadot} +2i\psi^\C\xi_{\C A}\frac{\partial}{\partial \rhobar^{\alphadot\betadot}}+i \psi_A \xibar^\gammadot_\alphadot \frac{\partial}{\partial \rhobar^{\gammadot\betadot}}+\dots,\\
 \tilde{Y}_A = \xitilde\frac{\partial}{\partial\psi^A} -\frac{4}{3}\psibar^\betadot\xitilde_\betadot^\B\frac{\partial}{\partial \rho^{A\B}}, \quad   \tilde{\overbar{Y}}_\alphadot = \xitilde\frac{\partial}{\partial\psibar^\alphadot} -\frac{4}{3}\psi^\B\xitilde_\B^\betadot\frac{\partial}{\partial \rhobar^{\alphadot\betadot}},
\end{gathered}
\end{equation}
where  $\dots$ denote appropriate symmetrization. This completes the proof of Lemma \bref{lemma:kernel} as prolongations of the above vector fields generate the required distribution $\bar\cK^\prime$.
\end{appendix}

\providecommand{\href}[2]{#2}\begingroup\raggedright\endgroup

\end{document}

%% file: my-heading.tex
\usepackage[mathscr]{eucal}
\usepackage{amsmath,amsfonts,amssymb,amsthm}
\usepackage{times}
\usepackage{hyperref}

\numberwithin{equation}{section} 

\newtheorem{thm}{Theorem}[section]

 \newtheorem{prop}[thm]{Proposition}
 \newtheorem{lemma}[thm]{Lemma}

\renewcommand{\tilde}{\widetilde}
\renewcommand{\hat}{\widehat}

\renewcommand{\simeq}{\cong}
\newcommand{\bref}[1]{\textbf{\ref{#1}}}

\newcommand{\p}[1]{|#1|}
\newcommand{\gh}[1]{\mathrm{gh}(#1)}

\newcommand{\dx}{\mathrm{d}_X}

\renewcommand{\d}{\partial}

\newcommand{\cF}{\mathcal{F}}

\newcommand{\tensor}{\otimes}

\newcommand{\binner}[2]{%
  {\langle}\kern-4.15pt{\langle}#1{,}\,#2{\rangle}\kern-4.15pt{\rangle}}
\newcommand{\commut}[2]{[#1{,}\,#2]}

\newcommand{\half}{\mathchoice{%
    \ffrac{1}{2}}{\frac{1}{2}}{\frac{1}{2}}{\frac{1}{2}}}

\newcommand{\ffrac}[2]{\raisebox{.5pt}%
  {\footnotesize$\displaystyle\frac{#1}{#2}$}\kern1pt}

\newcommand{\dl}[1]{\mathchoice{\ffrac{\d}{\d #1}}{\frac{\d}{\d #1}}{\ffrac{\d}{\d #1}}{\ffrac{\d}{\d #1}}}

\newcommand{\st}[2]{{\overset{#1}{#2}}}

\newcommand{\ddl}[2]{\ffrac{\d #1}{\d #2}}

\def\const{\mathop\mathrm{const}\nolimits}

\newcommand{\manifold}[1]{\mathscr{#1}}
\newcommand{\manX}{\manifold{X}}

\newcommand{\Liealg}{\mathfrak} 
\newcommand{\algg}{\Liealg{g}}

\newcommand{\algA}{\mathcal{A}}

\newcommand{\cC}{\mathcal{C}}

\newcommand{\fR}{\mathbb{R}}

\def\cK{\mathcal{K}}

 \def\cP{\mathcal{P}}
 \def\cT{\mathcal{T}}
 
\def\tr{{\rm Tr}}

\newcommand\blfootnote[1]{%
  \begingroup
  \renewcommand\thefootnote{}\footnote{#1}%
  \addtocounter{footnote}{-1}%
  \endgroup
}